\documentclass[11pt,a4paper]{article}
\usepackage[hyphens]{url}  
\usepackage{graphicx} 
\usepackage{natbib}  
\usepackage{caption} 
\usepackage{algorithm}
\usepackage{algorithmic}
\usepackage{vmargin}
\usepackage{hyperref}

\usepackage{newfloat}
\usepackage{listings}

\usepackage{authblk}
\usepackage{amsthm}
\usepackage{amsfonts} 
\usepackage{amssymb}
\usepackage{xcolor}
\usepackage{mathtools}
\usepackage{cleveref}
\usepackage{comment}
\usepackage{complexity}
\usepackage{xspace}
\usepackage{todonotes}

\usepackage{tikz,pgf}
\usetikzlibrary{positioning}
\usetikzlibrary{calc}
\usepackage{graphicx}

\newcommand{\MAPF}{\textsc{Multiagent Path Finding}\xspace}
\newcommand{\MAPFShort}{\textsc{MAPF}\xspace}
\newcommand{\MAPFMA}{\textsc{Multiagent Path Finding with Malfunctioning Agents}\xspace}
\newcommand{\MAPFMAShort}{\textsc{MAPFMA}\xspace}
\newcommand{\scheduleLength}{\ensuremath{\mu}}

\newcommand{\protocolOne}{\textsc{Check Before Moving}\xspace}
\newcommand{\protocolOneShort}{\textsc{CBM}\xspace}
\newcommand{\protocolTwo}{\textsc{Upgraded Check Before Moving}\xspace}
\newcommand{\protocolTwoShort}{\textsc{UCBM}\xspace}
\newcommand{\protocolThree}{\textsc{Check Counter Before Moving}\xspace}
\newcommand{\protocolThreeShort}{\textsc{CCBM}\xspace}

\makeatletter
\renewcommand*{\ALG@name}{Protocol}
\makeatother

\newtheorem{theorem}{Theorem}
\newtheorem{example}{Example}
\newtheorem{claim}{Claim}
\newenvironment{proofclaim}{\noindent{\em Proof of the claim.}}{\qedclaim}
\newcommand{\qedclaim}{\hfill $\diamond$ \medskip}

\newif\iflong
\newif\ifshort

\longtrue

\iflong
\else
\shorttrue
\fi

\begin{document}

\title{When Agents Break Down in Multiagent Path Finding}

\author[1]{Foivos Fioravantes}
\author[1]{Dušan Knop}
\author[2]{Nikolaos Melissinos}
\author[1]{Michal Opler}
\affil[1]{Department of Theoretical Computer Science, Faculty of Information Technology, Czech Technical University in Prague, Prague, Czech Republic}
\affil[2]{Computer Science Institute, Faculty of Mathematics and Physics, Charles University}
\date{}
\maketitle

\begin{abstract}
In \MAPF (\MAPFShort{}), the goal is to compute efficient, collision-free paths for multiple agents navigating a network from their sources to targets, minimizing the schedule’s makespan-the total time until all agents reach their destinations. We introduce a new variant that formally models scenarios where some agents may experience delays due to malfunctions, posing significant challenges for maintaining optimal schedules.

Recomputing an entirely new schedule from scratch after each malfunction is often computationally infeasible. To address this, we propose a framework for dynamic schedule adaptation that does not rely on full replanning. Instead, we develop protocols enabling agents to locally coordinate and adjust their paths “on the fly.” We prove that following our primary communication protocol, the increase in makespan after 
$k$ malfunctions is bounded by $k$ additional turns, effectively limiting the impact of malfunctions on overall efficiency. Moreover, recognizing that agents may have limited computational capabilities, we also present a secondary protocol that shifts the necessary computations onto the network’s nodes, ensuring robustness without requiring enhanced agent processing power. Our results demonstrate that these protocols provide a practical, scalable approach to resilient multiagent navigation in the face of agent failures.

\end{abstract}

\section{Introduction}

The classical \MAPF{} (\MAPFShort) problem focuses on computing non-colliding paths for a set of $k$ agents. Finding an optimal schedule—minimizing either the makespan or the total energy (i.e., the total distance traveled by all agents) -- is known to be computationally hard~\cite{ma2022graph,FKKMO24,FKKMOV25}. Even heuristic approaches struggle with scalability and performance~\cite{SharonSFS15}. In this work, we initiate the study of \MAPF{} in the presence of faulty agents. Specifically, we consider scenarios where a scheduled agent may unexpectedly pause its movement, incurring a delay of up to $d \in \mathbb{N}$ time steps. Such interruptions can invalidate the original schedule by introducing collisions. Two obvious remedies are: (1) recomputing the entire schedule from the agents’ current positions, or (2) pausing the execution of all agents until the delayed one can resume. While both approaches resolve conflicts, they come with significant drawbacks: recomputation may be computationally infeasible in real time, and global pausing requires coordinated communication among all agents. Our goal is to analyze the effects of such malfunctions and propose decentralized recovery strategies that preserve as much of the original schedule as possible.

Traditional \MAPFShort{} algorithms assume perfectly reliable agents that follow their assigned schedules without deviation.
However, in real-world deployments -- such as warehouse robotics, autonomous drones, or multi-robot exploration -- agents may experience delays due to hardware faults, temporary obstacles, or energy constraints.
In such settings, requiring all agents to pause or triggering a global replanning procedure is often impractical due to latency, communication limitations, or computational overhead (think of, e.g., traffic control).
This motivates the need for local recovery mechanisms that allow the system to tolerate occasional agent delays with minimal disruption.
We propose a framework for decentralized schedule repair, where only agents directly affected by a delay participate in resolving potential conflicts.
Our approach leverages the structure of the original schedule and local interactions to avoid global rescheduling, enabling more resilient and scalable multi-agent coordination.

To this end, we propose the \MAPFMA{} problem, where agents may pause for one or more turns. We observe that even with a single turn delay, conflicts can arise if other agents stick to the original schedule, especially under worst-case, adversarial conflict resolution—potentially leading to scenarios with no feasible solutions (see Example~\ref{ex:infiniteDelay}).

In this work, we analyze protocols to propagate information about delays. The most straightforward approach is to broadcast delay information and to suspend all agents until the malfunctioning agent resumes. Although this guarantees a conflict-free schedule, it introduces a uniform delay for all agents, which is often inefficient in practice. Furthermore, broadcasting may not be feasible in all settings, and such a protocol can unnecessarily delay unaffected agents.

\subsection{Our Results}
We begin by observing that adapting a given schedule in a centralized way is infeasible. In particular, we show deciding in a centralized way how to adapt the initial schedule so that its makespan is unaffected, is computationally hard (see Theorem~\ref{thm:keeping:the:same:makespan:is:NP:hard}). 

We then focus on the case of a single delay (i.e., a delay of exactly one turn), we propose the 
\protocolOne protocol, a natural and intuitive approach that prioritizes agents that have already experienced a delay. Specifically, this protocol grants the delayed agent priority access to the next vertex, ensuring conflict resolution without requiring global coordination. Surprisingly, as we show in Theorem~\ref{thm:delay-one}, this simple strategy guarantees a conflict-free schedule and increases the overall makespan by only one turn, which is clearly necessary (e.g., if the delayed agent is at distance equal to the original makespan from its target).

Moreover, this protocol is highly localized: before entering a vertex, an agent inspects its immediate neighborhood. If a delayed agent is detected in the vicinity, the agent pauses its own movement and signals the delay in the following turn. This localized mechanism minimizes communication overhead and leverages only local interactions to resolve conflicts.

Extending this protocol to delays of $k$ turns introduces further complexity. In such cases, the delay signal must propagate up to $k$ hops, and agents may need to scan a neighborhood of radius $k$ to respond appropriately. Although conceptually simple, this requirement implies that agents must be explicitly designed to handle delays up to $k$. Unexpected longer delays (e.g. $k+1$) can break the protocol guarantees, leading to potential deadlocks or unsafe behaviors.

To address such robustness concerns, we introduce a second protocol inspired by the ``Hansel and Gretel'' paradigm. Here, agents leave markers (such as pebbles or stickers) on vertices they visit. Each agent knows how many agents are expected to traverse each vertex on its planned path, and it scans only its adjacent vertex to determine when it may safely proceed. This lightweight marking system encodes minimal but sufficient information along the path and is based solely on local observations.

Despite its simplicity, the second protocol ensures a makespan increase of at most~$k$, in the presence of $k$ faulty agents, as we show in Theorem~\ref{thm:protocol2}. Moreover, we believe that it offers a practical balance between robustness and minimal local sensing, making it suitable for real-world multiagent systems.

\subsection{Related Work}
\MAPF (\MAPFShort) is a fundamental and practically motivated problem with a wide variety of algorithmic frameworks and real-world applications~\cite{ma2022graph}; see also surveys~\cite{stern2019,amigoni2022}.

The idea of reusing previously computed solutions has been a subject of interest across planning domains. In fact, early work in general planning showed that reusing a plan can, paradoxically, be more complex than computing a new one from scratch~\cite{NebelK95}. 

Motivated by real-world uncertainty, several variants of \MAPFShort{} have been proposed to handle malfunctions or execution delays. One notable model is $k$-\textsc{Robust} \MAPFShort, where the solution must tolerate up to $k$ delays at arbitrary times and agents~\cite{AtzmonSFWBZ20}. Another related model assumes that agents occupy multiple vertices simultaneously (e.g. due to physical size), affecting space-time occupancy constraints~\cite{Li_Surynek_Felner_Ma_Kumar_Koenig_2019}. These models are typically solved using adaptations of Conflict-Based Search (CBS)~\cite{SharonSFS15}, which remains a central technique in \MAPFShort{} literature.

A similar problem is the so-called conformant planning~\cite{CimattiRB04,Bonet10,HoffmannB06}, where the source of each agent cannot be accurately measured.
Due to the inherent hardness of these problems and the impossibility to predict delays and malfunctions, many researchers focus on probabilistic solutions that work well in expectation or with high probability~\cite{MaKK17,AtzmonSFSK20}.
Similarly to us, recently, Kottinger et al.~\cite{KottingerGASL24} studied how to add delays to a schedule so that it becomes conflict-free after some agent is delayed.
They discovered that computing the minimum number of delays is \textsf{APX-hard}.
This is in a strong contrast with our work where we focus on the global makespan and not on the sum over all agents.

As was already pointed out~\cite{AtzmonSFWBZ20,NebelK95} it may be even more challenging to find a solution that reuses as much as possible from an old solution.
This was previously observed in some other settings, for example, in the case of stable marriage and stable roommates~\cite{BredereckCKLN20}, where it turns out to be challenging to find a stable matching that is close to the original (after one agent ``resolves'' a single tie in their preferences).
We can observe the same phenomenon in graph representations where it might be significantly harder to find a representation as some part of the desired representation is given along with the input~\cite{chaplick2014contact,chaplick2018partial}.

Another aspect that is close to our work is the one of distributed computing~\cite{garg2002elements,kshemkalyani2011distributed}.
This was previously studied in the context of \MAPF{}~\cite{Pianpak19}.
Ma et al.~\cite{maLM21} investigated distributed coordination and the effect of communication.

\section{Preliminaries}
We begin by formally defining the \MAPF{} problem. As input, we receive a graph $G=(V,E)$, a set of agents~$A$, two functions $s_0\colon A \rightarrow V$, $t\colon A \rightarrow V$ and a positive integer~$\ell$, known as the \textit{makespan}. 
For any pair $a , b \in A$ where $a \neq b$, we have that $s_0(a) \neq s_0(b)$ and $t(a) \neq t(b)$.
Initially, each agent $a \in A$ is placed on $s_0(a)$.
At specific time intervals, called \textit{turns}, the agents are allowed to move to a vertex neighboring their position in the previous turn, without being obliged to do so.
The agents can make at most one move per turn, and each vertex can host at most one agent at a given turn.
The \textit{intended position} of the agents at the end of the turn $i$ (after the agents have moved) is given by an injective function $s_i\colon A \rightarrow V$, $i\in[\scheduleLength]$.
Any sequence $\sigma=s_0, s_1, \ldots, s_\scheduleLength$ is called a \textit{schedule} of \textit{length} $\scheduleLength$. If, moreover, $\sigma$ also respects the above rules (collision free), then we call it a \textit{feasible schedule}.
We also briefly mention that there are two main variants of the classical \MAPFShort{} problem, according to whether swaps are allowed or not. A \textit{swap} is the behavior when two agents start from adjacent positions and exchange them within one turn. In this work, we do not allow swaps. Given an instance $\langle G, A, s_0, t, \ell\rangle$ of \MAPFShort{}, we are tasked with finding a feasible schedule of length $\ell$.

In this work, we introduce the \MAPFMA{} 
(\MAPFMAShort{}) problem, a version of \MAPFShort{} that aims to capture scenarios where agents may malfunction. 
Before we formally define this problem, we first need to define what we mean by malfunctioning and, more generally, delaying. 

Given a schedule $\sigma$ of length $\ell$, we will say that the agents of $A'\subseteq A$ perform (or participate in) a \textit{delay-$1$} operation in turn $i$ if each agent $a\in A'$ follows a new schedule $\sigma'(a)=s'_0(a),\ldots,s'_{\ell+1}(a)$ such that:
\begin{itemize}
    \item $s'_j(a) = s_j(a)$, for all $j <i$, and 
    \item $s'_j(a) = s_{j-1}(a)$, for all $j \ge i$.
\end{itemize}
To improve the flow, we may also say that the set $A'$ \textit{delays-$1$}. In essence, the agents of $A'$ will follow the initial schedule, with the difference that they will spend one extra turn at one vertex (each agent on a different vertex). 
We now have two sets of agents, the ones who follow $\sigma$ and the ones who follow $\sigma'$. Technically, this poses an issue on the definition of the notion of schedule for all the agents of $A$. To overcome this, we amortize the schedules $\sigma$ and $\sigma'$:
\begin{itemize}
    \item If it is true that for all $a\in A'$ we have $s'_{\ell}(a) = s'_{\ell+1}(a)$, then we prune the $\ell+1$ turns and set $\sigma'=s'_0(a),\ldots,s'_{\ell}(a)$ for every $a\in A'$.
    \item Otherwise, there exists an agent $a\in A'$ such that $s'_{\ell}(a) \neq s'_{\ell+1}(a)$. In this case, we extend the schedules of all the agents $b \in A\setminus A'$ by setting $s_{\ell+1}(b) = s_{\ell}(b)$.
\end{itemize}
So, if a set $A'$ of agents performs a delay-$1$ upon the given schedule $\sigma$, then the agents of $A$ will follow the amortized schedule $\sigma'$ where $\sigma'(a)=\sigma(a)$ if $a\in A\setminus A'$ (including the $\ell+1$ turn if we are in the case of the second item above) and $\sigma'(a)= s'_0(a),\ldots,s'_{\ell+1}(a)$ if $a\in A'$ (excluding $\ell +1$ turn if are in the case of the first item above). To avoid cumbersome notations, we will henceforth assume that the resulting schedule $\sigma'$ is always amortized, unless mentioned otherwise.

Finally, we can model the scenario where some agents introduce multiple delays. This is done by recursively performing the delay-$1$ operation the corresponding number of times. Formally, we define the sequence of turns where the delays are introduced and the sets of agents that introduce a delay in each of the given turns. We assume that the delay introductions are given in the order that are (actually) happening. That is, a delay introduction in turn $i$ must appear before any delay introduction in any turn $j>i$.
In order to compute the schedule after a set of $k\ge 1$ delay-$1$ operations, we compute a sequence of schedules $\sigma^0,\sigma^1,\ldots,\sigma^k$ such that: 
\begin{itemize}
    \item $\sigma^0=\sigma$, where $\sigma$ is the initial given schedule, and 
    \item for any $i\in [k]$, the schedule $\sigma^{i}$ is obtained by applying the $i^{th}$ delay-$1$ to the schedule $\sigma^{i-1}$.
\end{itemize}

We would like to comment on two details at this point. First, performing the delay-$1$ operation upon a \textit{feasible} schedule $\sigma$ may very well result in a schedule $\sigma'$ that is no longer feasible. Second, this operation is not to be identified as the corresponding agents malfunctioning. Instead, it is meant as an operation that models deviation. In other words, any agent, malfunctioning or not, may participate in such an operation. The distinction has to do with why they do so. Indeed, the malfunctioning agents have no choice, and at some point will participate in a delay-$1$ operation, while the non-malfunctioning agents will participate in such an operation to ensure the feasibility of the resulting schedule. To avoid confusion, we may use the term \textit{malfunction-$1$} operation to denote the former case. 

We are now able to formally define \MAPFMAShort{}. As input we receive an instance $\mathcal{I}=\langle G, A, s_0, t, \ell\rangle$ of \MAPFShort, a feasible schedule $\sigma$ of length $\ell$ and two integers $k\geq 1$ and $\ell'\geq \ell$. Moreover, there is an adversary that can force $k$ malfunction-$1$ operations on the agents of $A$. Our task is to construct a feasible schedule $\sigma'$ of length $\ell'$ that includes the malfunction-$1$ operations imposed by the adversary. 

\section{No Communication}

Recall that in the \MAPFMAShort{} problem, our task is to orchestrate a set of delays, so that the resulting schedule is affected in the least way possible by the malfunctioning of one agent. In this section we argue that this orchestration is indeed crucial. In essence, we study the simplest protocol one can propose to adapt the schedule after some agents has malfunctioned: ``ignore'' the malfunctioning agents. Simply put, all the agents will try to continue traversing the network as if no malfunction has happened. In this case, it could be that (at least) two agents try to access the same vertex during the same turn. During such a turn, we give the power to assign priorities to an adversary. 

Although such a protocol would require small effort from the implementation point of view, we show that it can lead to catastrophic scenarios. 


\begin{example}[Example of infinite delay]\label{ex:infiniteDelay}
\end{example}
Consider the instance shown in Figure~\ref{fig:catastrophy-1}. Let $a_1$ and $a_2$ be the agents depicted in red and blue respectively. Clearly, the makespan in this schedule is at least $2$, and there is a (unique) schedule achieving this: $s_1(a_1)=u_1$, $s_1(a_2)=u_2$, $s_2(a_1)=u_2$, $s_2(a_2)=u_3$. Consider now what happens if $a_2$ malfunctions-$1$ at turn $1$. Then $s'_1(a_2)=u_4$,$s'_2(a_2)=u_2$ and $s'_3(a_2)=u_3$. But we also have that $s'_2(a_1)=u_2$. At this point the adversary would choose that agent $a_1$ moves to $u_2$ on the second turn, where they would stay forever. That is $s'_1(a_1)=s_1(a_1)=u_1$, $s'_2(a_1)=u_2$ and $s'_3(a_1)=u_2$. But now $u_2$ can never be reached by $a_2$, and thus $a_2$ can never reach their terminal. To sum up, with just one agent malfunctioning during one turn, we went from a schedule that had a length of $2$ to an instance that admits no feasible schedule. 

\begin{figure}[!t]
\centering

\begin{tikzpicture}[scale=0.45, inner sep=0.7mm]
\node[draw, circle, line width=1pt, fill=white](u1) at (0,0)[label=below: $u_1$] {};
\node[draw, circle, line width=1pt, fill=red](u2) at (2,0)[label=below: $u_2$] {};
\node[draw, circle, line width=1pt, fill=blue](u3) at (4,0)[label=below: $u_3$] {};
\node[draw, circle, line width=1pt, fill=white](u4) at (2,2)[label=right: $u_4$] {};

\node[above = 0.1 of u1.west] () {\includegraphics[scale=0.07]{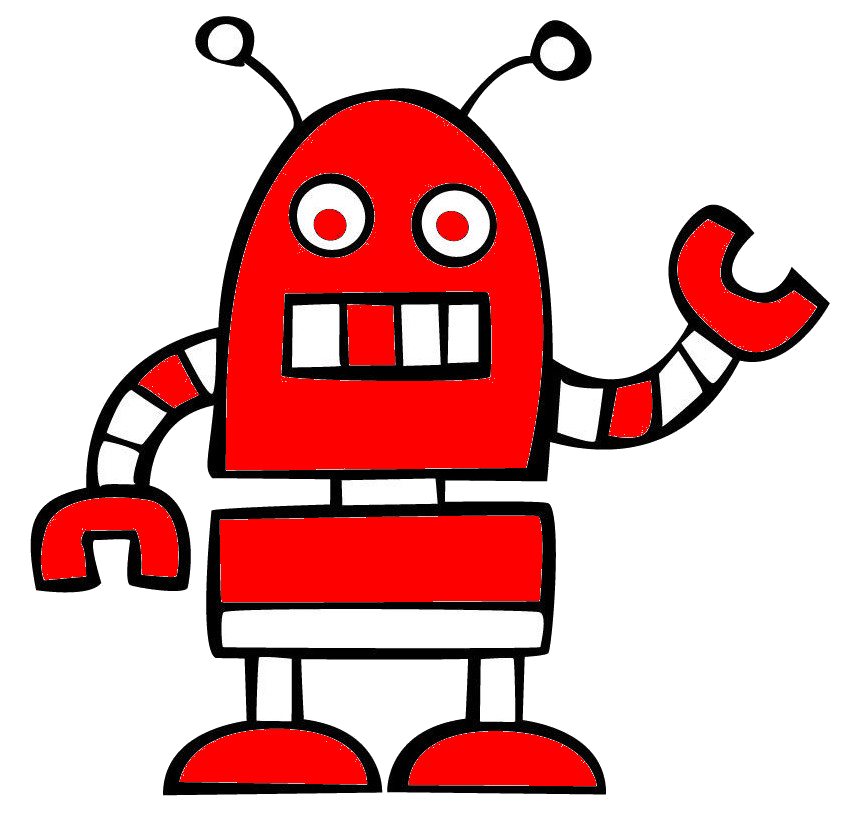}};
\node[above = 0.1 of u4.west] () {\includegraphics[scale=0.07]{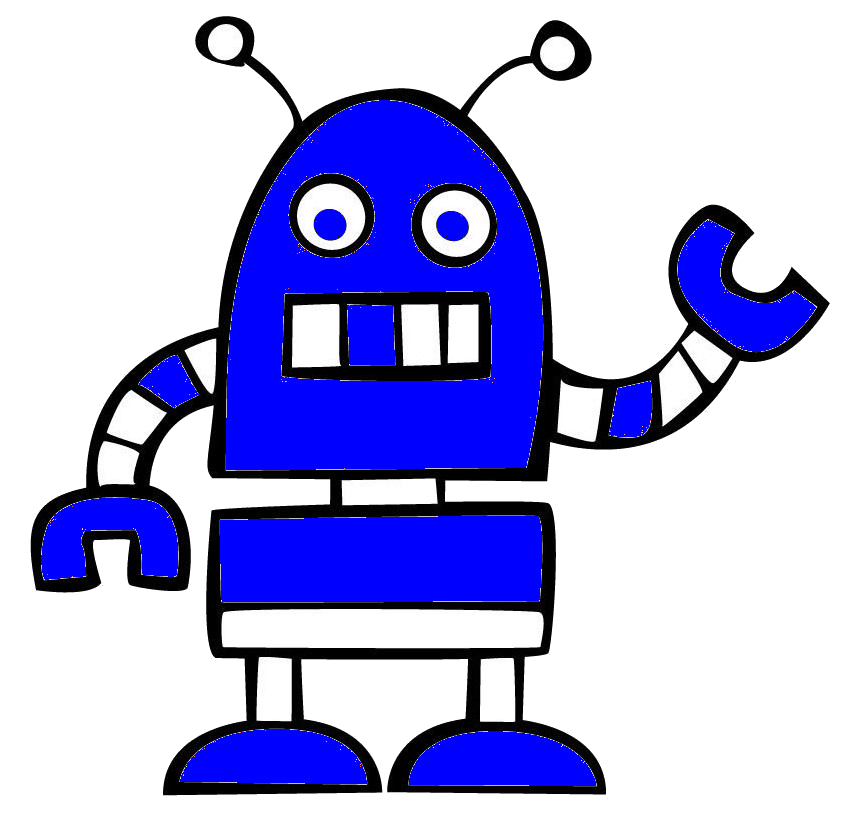}};

\draw[-, line width=0.5pt]  (u1) -- (u2);
\draw[-, line width=0.5pt]  (u2) -- (u3);
\draw[-, line width=0.5pt]  (u4) -- (u2);
\end{tikzpicture}
\caption{The initial state of the first example of catastrophic scenario. The colors on the agents and the vertices are used to encode the terminal vertex of each agent.}\label{fig:catastrophy-1}
\end{figure}
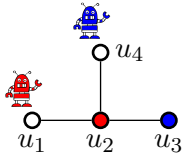

\smallskip 

However, it can be argued that the situation depicted in Figure~\ref{fig:catastrophy-1} is fairly straightforward to be identified and dealt with. In the next example, we establish that there can be far more nuanced problematic cases. 

\begin{example}[Example of large delay]
\end{example}
Consider the scenario depicted in Figure~\ref{fig:catastrophy-2}. We will say that the red blue and brown agents are the \textit{colored} ones; each colored agent has as terminal the vertex colored by their color. Let $a_1,a_2$ and $a_3$ denote the red, blue, and brown agents, respectively. The gray vertices are called \textit{critical}; let $c_1,c_2$ and $c_3$ denote the critical vertices that are neighbors of $v_1,v_2$ and $v_3$, respectively. Any path that starts from black vertices, goes through the closest critical vertex and then goes on to reach the vertices with the closest black agents is called an \textit{auxiliary} path. Lets focus on one auxiliary path, say the one that crosses $c_3$. The black agents of this path have as terminals the black vertices of the same path: the bottom (top resp.) black agent has as terminal the left (right resp.) black vertex of this path. The terminals of the other black agents are defined analogously. 

Given a feasible schedule $\sigma$ of makespan $\ell$ for this instance, we claim that there is a way to perform a delay-$1$ operation on one agent which results in any feasible schedule $\sigma'$ having length at least $\ell+2$. 

Let us first consider $\sigma$. It is clear that $\ell\geq 9$, as the colored agents are at distance $9$ from their terminals. We can actually design a schedule that achieves this makespan, and thus is optimal: 
\begin{itemize}
    \item For every turn $i\in[9]$, every colored agent moves towards their terminal following the unique shortest path that leads them there.
    \item For every turn $i\in[5]$ the black agents move towards their terminals on their respective auxiliary paths. 
    \item In turn $i=6$ the black agents wait on their current vertices (the two vertices preceding the corresponding critical vertex in each auxiliary path). 
    \item For every turn $i\in[7,9]$ the black agents continue towards their terminals. 
\end{itemize}
This is clearly a feasible schedule. Observe first that the colored agents are on the triangle during the turns $4$ and $5$. In particular, we have that $s_4(a_1)=v_1$, $s_4(a_2)=v_2$, $s_4(a_3)=v_3$ and $s_5(a_1)=v_3$, $s_5(a_2)=v_1$, $s_5(a_3)=v_2$. Observe also that the critical vertices are occupied by the black agents during the turns $7$ and $8$. 

\begin{figure}[!t]
\centering

\begin{tikzpicture}[scale=0.45, inner sep=0.6mm]

\node[draw, circle, line width=1pt, fill=white](t1) at (4,4)[label=below: $v_1$] {};
\node[draw, circle, line width=1pt, fill=white](t2) at (7,4)[label=right: $v_2$] {};
\node[draw, circle, line width=1pt, fill=white](t3) at (5.5,7)[label=right: $v_3$] {};

\draw[-, line width=0.5pt]  (t1) -- (t2);
\draw[-, line width=0.5pt]  (t2) -- (t3);
\draw[-, line width=0.5pt]  (t3) -- (t1);

\node[draw, circle, line width=1pt, fill=white](l1) at (0,0)[] {};
\node[draw, circle, line width=1pt, fill=white](l2) at (1,0)[] {};
\node[draw, circle, line width=1pt, fill=white](l3) at (2,0)[] {};
\node[draw, circle, line width=1pt, fill=white](l4) at (3,0)[] {};
\node[draw, circle, line width=1pt, fill=white](l5) at (3,1)[] {};
\node[draw, circle, line width=1pt, fill=white](l6) at (3,2)[] {};
\node[draw, circle, line width=1pt, fill=white](l7) at (3,3)[] {};
\node[draw, circle, line width=1pt, fill=gray!70](l8) at (3,4)[] {};
\node[draw, circle, line width=1pt, fill=black](l9) at (3,5)[] {};
\node[draw, circle, line width=1pt, fill=black](l10) at (3,6)[] {};
\node[draw, circle, line width=1pt, fill=blue](l11) at (0,4)[] {};
\node[draw, circle, line width=1pt, fill=white](l12) at (1,4)[] {};
\node[draw, circle, line width=1pt, fill=white](l13) at (2,4)[] {};

\draw[-, line width=0.5pt]  (l1) -- (l2);
\draw[-, line width=0.5pt]  (l2) -- (l3);
\draw[-, line width=0.5pt]  (l3) -- (l4);
\draw[-, line width=0.5pt]  (l4) -- (l5);
\draw[-, line width=0.5pt]  (l5) -- (l6);
\draw[-, line width=0.5pt]  (l6) -- (l7);
\draw[-, line width=0.5pt]  (l7) -- (l8);
\draw[-, line width=0.5pt]  (l8) -- (l9);
\draw[-, line width=0.5pt]  (l9) -- (l10);
\draw[-, line width=0.5pt]  (l11) -- (l12);
\draw[-, line width=0.5pt]  (l12) -- (l13);
\draw[-, line width=0.5pt]  (l13) -- (l8);
\draw[-, line width=0.5pt]  (l8) -- (t1);

\node[above = 0.1 of l11.west] () {\includegraphics[scale=0.07]{robot-red.png}};
\node[above = 0.1 of l1.west] () {\includegraphics[scale=0.06]{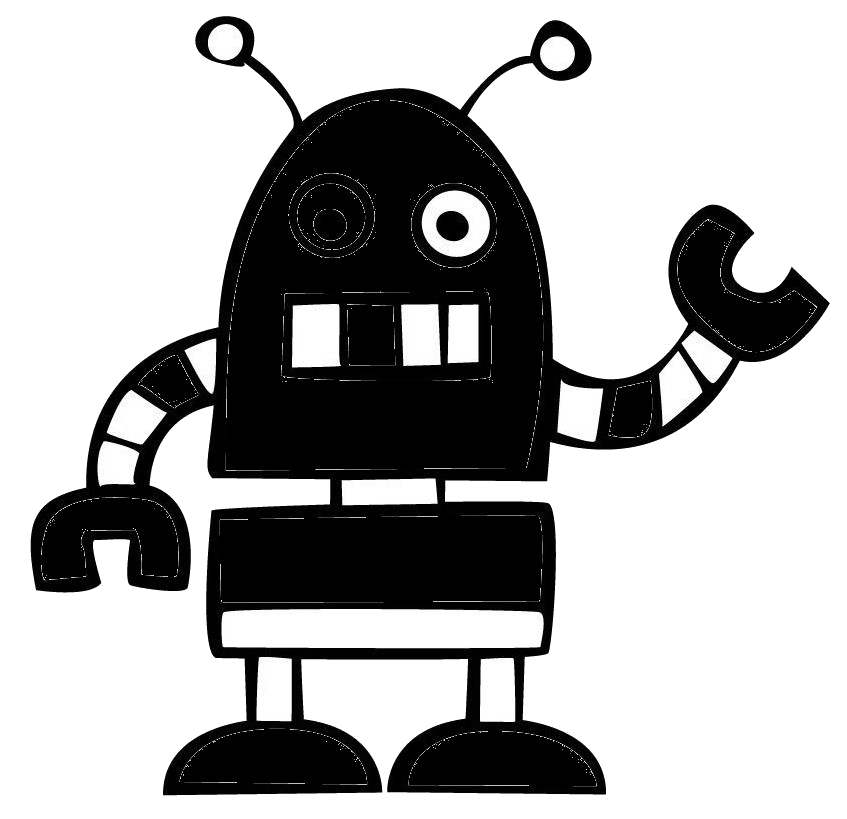}};
\node[above = 0.1 of l2.west] () {\includegraphics[scale=0.06]{robot-black.png}};


\node[draw, circle, line width=1pt, fill=white](r1) at (11,0)[] {};
\node[draw, circle, line width=1pt, fill=white](r2) at (11,1)[] {};
\node[draw, circle, line width=1pt, fill=white](r3) at (11,2)[] {};
\node[draw, circle, line width=1pt, fill=white](r4) at (11,3)[] {};
\node[draw, circle, line width=1pt, fill=white](r5) at (10,3)[] {};
\node[draw, circle, line width=1pt, fill=white](r6) at (9,3)[] {};
\node[draw, circle, line width=1pt, fill=white](r7) at (8,3)[] {};
\node[draw, circle, line width=1pt, fill=gray!70](r8) at (7,3)[] {};
\node[draw, circle, line width=1pt, fill=black](r9) at (6,3)[] {};
\node[draw, circle, line width=1pt, fill=black](r10) at (5,3)[] {};
\node[draw, circle, line width=1pt, fill=brown](r11) at (7,0)[] {};
\node[draw, circle, line width=1pt, fill=white](r12) at (7,1)[] {};
\node[draw, circle, line width=1pt, fill=white](r13) at (7,2)[] {};

\draw[-, line width=0.5pt]  (r1) -- (r2);
\draw[-, line width=0.5pt]  (r2) -- (r3);
\draw[-, line width=0.5pt]  (r3) -- (r4);
\draw[-, line width=0.5pt]  (r4) -- (r5);
\draw[-, line width=0.5pt]  (r5) -- (r6);
\draw[-, line width=0.5pt]  (r6) -- (r7);
\draw[-, line width=0.5pt]  (r7) -- (r8);
\draw[-, line width=0.5pt]  (r8) -- (r9);
\draw[-, line width=0.5pt]  (r9) -- (r10);
\draw[-, line width=0.5pt]  (r11) -- (r12);
\draw[-, line width=0.5pt]  (r12) -- (r13);
\draw[-, line width=0.5pt]  (r13) -- (r8);
\draw[-, line width=0.5pt]  (r8) -- (t2);

\node[right = 0.2 of r11.west] () {\includegraphics[scale=0.07]{robot-blue.png}};
\node[right = 0.2 of r1.west] () {\includegraphics[scale=0.06]{robot-black.png}};
\node[right = 0.2 of r2.west] () {\includegraphics[scale=0.06]{robot-black.png}};


\node[draw, circle, line width=1pt, fill=white](u1) at (9.5,11)[] {};
\node[draw, circle, line width=1pt, fill=white](u2) at (9.5,10)[] {};
\node[draw, circle, line width=1pt, fill=white](u3) at (9.5,9)[] {};
\node[draw, circle, line width=1pt, fill=white](u4) at (9.5,8)[] {};
\node[draw, circle, line width=1pt, fill=white](u5) at (8.5,8)[] {};
\node[draw, circle, line width=1pt, fill=white](u6) at (7.5,8)[] {};
\node[draw, circle, line width=1pt, fill=white](u7) at (6.5,8)[] {};
\node[draw, circle, line width=1pt, fill=gray!70](u8) at (5.5,8)[] {};
\node[draw, circle, line width=1pt, fill=black](u9) at (4.5,8)[] {};
\node[draw, circle, line width=1pt, fill=black](u10) at (3.5,8)[] {};
\node[draw, circle, line width=1pt, fill=red](u11) at (5.5,11)[] {};
\node[draw, circle, line width=1pt, fill=white](u12) at (5.5,10)[] {};
\node[draw, circle, line width=1pt, fill=white](u13) at (5.5,9)[] {};

\draw[-, line width=0.5pt]  (u1) -- (u2);
\draw[-, line width=0.5pt]  (u2) -- (u3);
\draw[-, line width=0.5pt]  (u3) -- (u4);
\draw[-, line width=0.5pt]  (u4) -- (u5);
\draw[-, line width=0.5pt]  (u5) -- (u6);
\draw[-, line width=0.5pt]  (u6) -- (u7);
\draw[-, line width=0.5pt]  (u7) -- (u8);
\draw[-, line width=0.5pt]  (u8) -- (u9);
\draw[-, line width=0.5pt]  (u9) -- (u10);
\draw[-, line width=0.5pt]  (u11) -- (u12);
\draw[-, line width=0.5pt]  (u12) -- (u13);
\draw[-, line width=0.5pt]  (u13) -- (u8);
\draw[-, line width=0.5pt]  (u8) -- (t3);

\node[right = 0.2 of u11.west] () {\includegraphics[scale=0.07]{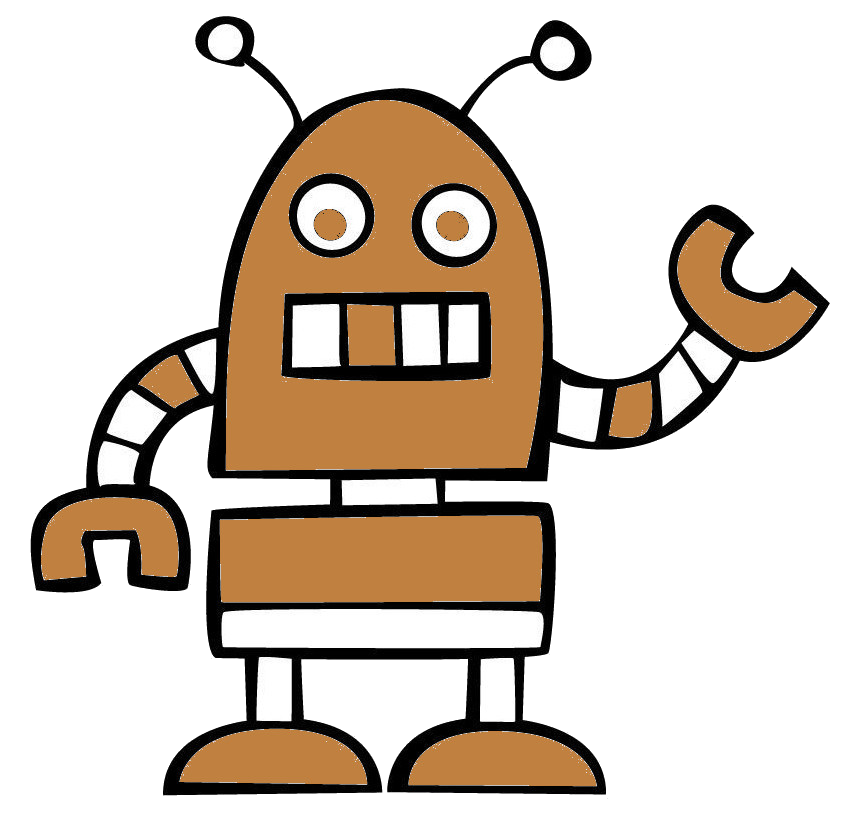}};
\node[right = 0.2 of u1.west] () {\includegraphics[scale=0.06]{robot-black.png}};
\node[right = 0.2 of u2.west] () {\includegraphics[scale=0.06]{robot-black.png}};

\end{tikzpicture}
\caption{The initial state of the second example of catastrophic scenario. The colors on the agents and the vertices are used to encode the terminal vertex of each agent. The color gray is used to denote the critical vertices.}\label{fig:catastrophy-2}
\end{figure}

Let us consider now what happens if $a_1$ performs a malfunction-$1$ operation in turn $4$. Then $s'_4(a_1)=s'_5(a_1)=v_1$. But we also have that $s'_5(a_2)=s_5(a_2)=v_1$. Thus, for $\sigma'$ to be feasible, we are obliged to have agent $a_2$ perform a delay-$1$ operation in turn $4$. With similar arguments, the same holds true for $a_3$. At this stage, we have that $s'_5(a_i)=v_i$ and $s'_6(a_i)=v_{i\mod{3}+1}$. Let us continue by focusing on the agent $a_1$ (the behavior of the other colored agents is symmetrical). We have that $s'_6(a_1)=c_3$. Let $b_1$ and $b_2$ denote the bottom and top, respectively, black agents of the auxiliary path that passes through $c_3$, and observe that $s'_6(b_1)=c_3$. At this point, the adversary may chose to prioritize $b_1$. This forces $a_1$ to perform another delay-$1$ operation in turn $6$. This situation is again repeated during turn $7$, this time with $a_1$ and $b_2$ both wanting to access $c_1$. Thus, $s'_7(a_1)=v_3$. After this point $a_1$ is free to move towards their target, which they will reach in turn $11$. 

To sum up, by having one agent malfunction for one turn, we ended up with a feasible schedule that has a length that is longer than the optimal one by $2$. Moreover, it is easy to see that for any $k\geq 2$, we can construct a scenario where having one agent malfunction for one turn would result in any feasible schedule having a length that is longer than the optimal by $k$. To achieve this, it suffices to carefully extend the auxiliary paths, the number of black agents (and black terminals) and the paths connecting the starting positions of the colored agents to the corresponding critical vertices.

\smallskip

Actually, the situation is even worse. Indeed, we are able to show that even acknowledging that there is a malfunction, deciding in a centralized way how to adapt the initial schedule so that its makespan is unaffected, is hopeless.  

\ifshort
\input{np-hardness-short}
\fi
\iflong

\begin{theorem}\label{thm:keeping:the:same:makespan:is:NP:hard}
    Let $\mathcal{I}=\langle G, A, s_0, t \rangle$ be an instance of \MAPFShort{} and $\sigma$ be a schedule of $\mathcal{I}$ length $\ell$. 
    It is \NP-hard to decide if there exists a set of delay-$1$ operations that can be performed after a malfunction-$1$ operation occurs, so that the resulting schedule is of length $\ell$. This holds even if $G$ is a planar graph of maximum degree $10$.
\end{theorem}

\begin{proof}

\begin{figure}[!t]
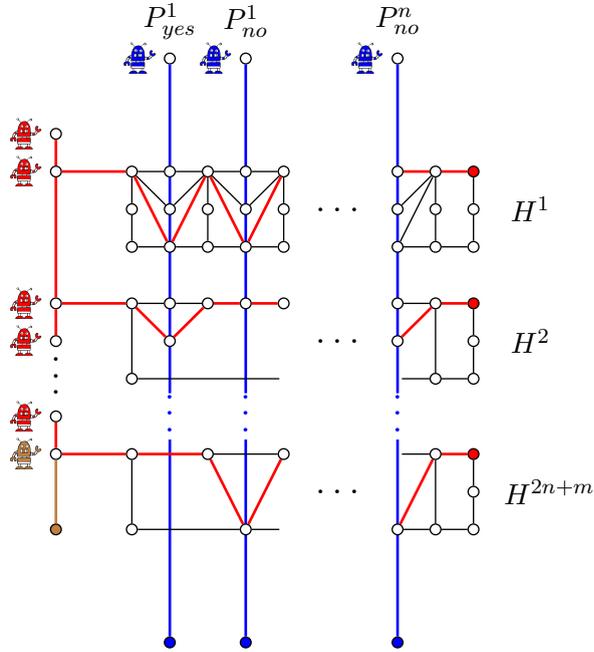

\centering

\begin{tikzpicture}[scale=0.5, inner sep=0.5mm]
\tikzstyle{vertex}=[draw, circle, line width=.5pt]
\tikzstyle{blueVertex}=[vertex, fill=blue, color=blue]

\node[vertex, fill=blue](e01) at (3,-3)[] {};
\node[draw, circle, line width=0.5pt, fill=blue](e02) at (5,-3)[] {};
\node[draw, circle, line width=0.5pt, fill=blue](e03) at (9,-3)[] {};

\node[](e11) at (3,2.5)[] {};
\node[](e12) at (5,2.5)[] {};
\node[](e13) at (9,2.5)[] {};

    \draw[-, line width=1pt, color=blue]  (e11) -- (e01);
    \draw[-, line width=1pt, color=blue]  (e12) -- (e02);
    \draw[-, line width=1pt, color=blue]  (e13) -- (e03);

\node[draw, circle, line width=0.5pt, fill=brown](u00) at (0,0)[] {};
\node[draw, circle, line width=0.5pt, fill=white](u01) at (2,0)[] {};
\node[draw, circle, line width=0.5pt, fill=white](u02) at (5,0)[] {};
\node[](u03) at (6,0)[] {};
\node[draw, circle, line width=0.5pt, fill=white](u04) at (9,0)[] {};
\node[draw, circle, line width=0.5pt, fill=white](u05) at (10,0)[] {};

\node[draw, circle, line width=0.5pt, fill=white](u10) at (0,2)[] {};
\node[draw, circle, line width=0.5pt, fill=white](u11) at (2,2)[] {};
\node[draw, circle, line width=0.5pt, fill=white](u12) at (4,2)[] {};
\node[draw, circle, line width=0.5pt, fill=white](u13) at (6,2)[] {};
\node[](u14) at (9,2)[] {};
\node[draw, circle, line width=0.5pt, fill=white](u15) at (10,2)[] {};

\draw[-, line width=1pt, color=brown]  (u00) -- (u10);
\draw[-, line width=0.5pt]  (u11) -- (u01);
\draw[-, line width=0.5pt]  (u01) -- (u02);
\draw[-, line width=0.5pt]  (u02) -- (u03);
\draw[-, line width=0.5pt]  (u04) -- (u05);
\draw[-, line width=0.5pt]  (u12) -- (u13);
\draw[-, line width=0.5pt]  (u14) -- (u15);
\draw[-, line width=0.5pt]  (u15) -- (u05);

\path (6,1) -- (9,1) node [black, midway, sloped] {\Large$\dots$};

\draw[-, line width=1pt, color=red]  (u10) -- (u11);
\draw[-, line width=1pt, color=red]  (u11) -- (u12);
\draw[-, line width=1pt, color=red]  (u12) -- (u02);
\draw[-, line width=1pt, color=red]  (u02) -- (u13);
\draw[-, line width=1pt, color=red]  (u04) -- (u15);

\node at (13,   1) {$H^{2n+m}$};

\node[draw, circle, line width=0.5pt, fill=white](u2) at (0,3)[] {};
\draw[-, line width=1pt, color=red]  (u2) -- (u10);

\node[left = 0.05 of u2] () {\includegraphics[scale=0.06]{robot-red.png}};
\node[left = 0.05 of u10] () {\includegraphics[scale=0.06]{robot-brown.png}};

\node[draw, circle, line width=0.5pt, fill=white](u06) at (11,0)[] {};
\node[draw, circle, line width=0.5pt, fill=white](w1) at (11,1)[] {};
\node[draw, circle, line width=0.5pt, fill=red](w2) at (11,2)[] {};

\draw[-, line width=0.5pt]  (u05) -- (u06);
\draw[-, line width=0.5pt]  (u06) -- (w1);
\draw[-, line width=0.5pt]  (w1) -- (w2);
\draw[-, line width=1pt, color=red]  (w2) -- (u15);

\node at (3, 13.5) {$P^1_{yes}$};
\node at (5, 13.5) {$P^1_{no}$};
\node at (9, 13.5) {$P^n_{no}$};

\node[draw, circle, line width=0.5pt, fill=white](b1) at (3,12.5)[] {};
\node[draw, circle, line width=0.5pt, fill=white](b2) at (5,12.5)[] {};
\node[draw, circle, line width=0.5pt, fill=white](b3) at (9,12.5)[] {};

\node[left = 0.05 of b1] () {\includegraphics[scale=0.06]{robot-blue.png}};
\node[left = 0.05 of b2] () {\includegraphics[scale=0.06]{robot-blue.png}};
\node[left = 0.05 of b3] () {\includegraphics[scale=0.06]{robot-blue.png}};

\node[](b4) at (3,3.5)[] {};
\node[](b5) at (5,3.5)[] {};
\node[](b6) at (9,3.5)[] {};
\draw[-, line width=1pt, color=blue]  (b1) -- (b4);
\draw[-, line width=1pt, color=blue]  (b2) -- (b5);
\draw[-, line width=1pt, color=blue]  (b3) -- (b6);

\path (b4) -- (e11) node [blue, midway, sloped] {\Large$\dots$};
\path (b5) -- (e12) node [blue, midway, sloped] {\Large$\dots$};
\path (b6) -- (e13) node [blue, midway, sloped] {\Large$\dots$};

\node[draw, circle, line width=0.5pt, fill=white](u20) at (2,4)[] {};
\node[](u21) at (6,4)[] {};
\node[](u22) at (9,4)[] {};
\node[draw, circle, line width=0.5pt, fill=white](u23) at (10,4)[] {};

\node[draw, circle, line width=0.5pt, fill=white](u30) at (3,5)[] {};
\node[draw, circle, line width=0.5pt, fill=white](u31) at (9,5)[] {};
\node[draw, circle, line width=0.5pt, fill=white](u40) at (0,6)[] {};
\node[draw, circle, line width=0.5pt, fill=white](u41) at (2,6)[] {};
\node[draw, circle, line width=0.5pt, fill=white](u42) at (3,6)[] {};
\node[draw, circle, line width=0.5pt, fill=white](u43) at (4,6)[] {};
\node[draw, circle, line width=0.5pt, fill=white](u44) at (5,6)[] {};
\node[draw, circle, line width=0.5pt, fill=white](u45) at (6,6)[] {};
\node[draw, circle, line width=0.5pt, fill=white](u46) at (9,6)[] {};
\node[draw, circle, line width=0.5pt, fill=white](u47) at (10,6)[] {};

\draw[-, line width=1pt, color=red]  (u40) -- (u41);
\draw[-, line width=1pt, color=red]  (u41) -- (u30);
\draw[-, line width=1pt, color=red]  (u30) -- (u43);
\draw[-, line width=1pt, color=red]  (u43) -- (u44);
\draw[-, line width=1pt, color=red]  (u44) -- (u45);
\draw[-, line width=1pt, color=red]  (u31) -- (u47);

\draw[-, line width=0.5pt]  (u41) -- (u20);
\draw[-, line width=0.5pt]  (u20) -- (u21);
\draw[-, line width=0.5pt]  (u22) -- (u23);
\draw[-, line width=0.5pt]  (u41) -- (u42);
\draw[-, line width=0.5pt]  (u42) -- (u43);
\draw[-, line width=0.5pt]  (u46) -- (u47);
\draw[-, line width=0.5pt]  (u47) -- (u23);

\node[draw, circle, line width=0.5pt, fill=white](u4) at (0,5)[] {};
\node[left = 0.05 of u4] () {\includegraphics[scale=0.06]{robot-red.png}};
\draw[-, line width=1pt, color=red]  (u4) -- (u40);
\path (u4) -- (u2) node [black, midway, sloped] {\Large$\dots$};
\node[left = 0.05 of u40] () {\includegraphics[scale=0.06]{robot-red.png}};
\path (6,5) -- (9,5) node [black, midway, sloped] {\Large$\dots$};
\node at (12.5,   5) {$H^2$};

\node[draw, circle, line width=0.5pt, fill=white](w3) at (11,4)[] {};
\node[draw, circle, line width=0.5pt, fill=white](w4) at (11,5)[] {};
\node[draw, circle, line width=0.5pt, fill=red](w5) at (11,6)[] {};

\draw[-, line width=0.5pt]  (u23) -- (w3);
\draw[-, line width=0.5pt]  (w3) -- (w4);
\draw[-, line width=0.5pt]  (w4) -- (w5);
\draw[-, line width=1pt, color=red]  (w5) -- (u47);

\node[draw, circle, line width=0.5pt, fill=white](u50) at (2,7.5)[] {};
\node[draw, circle, line width=0.5pt, fill=white](u51) at (3,7.5)[] {};
\node[draw, circle, line width=0.5pt, fill=white](u52) at (4,7.5)[] {};
\node[draw, circle, line width=0.5pt, fill=white](u53) at (5,7.5)[] {};
\node[draw, circle, line width=0.5pt, fill=white](u54) at (6,7.5)[] {};
\node[draw, circle, line width=0.5pt, fill=white](u55) at (9,7.5)[] {};
\node[draw, circle, line width=0.5pt, fill=white](u56) at (10,7.5)[] {};

\node[draw, circle, line width=0.5pt, fill=white](u60) at (2,8.5)[] {};
\node[draw, circle, line width=0.5pt, fill=white](u61) at (3,8.5)[] {};
\node[draw, circle, line width=0.5pt, fill=white](u62) at (4,8.5)[] {};
\node[draw, circle, line width=0.5pt, fill=white](u63) at (5,8.5)[] {};
\node[draw, circle, line width=0.5pt, fill=white](u64) at (6,8.5)[] {};
\node[draw, circle, line width=0.5pt, fill=white](u65) at (9,8.5)[] {};
\node[draw, circle, line width=0.5pt, fill=white](u66) at (10,8.5)[] {};

\node[draw, circle, line width=0.5pt, fill=white](u) at (0,10.5)[] {};
\node[draw, circle, line width=0.5pt, fill=white](u7) at (0,9.5)[] {};
\node[draw, circle, line width=0.5pt, fill=white](u70) at (2,9.5)[] {};
\node[draw, circle, line width=0.5pt, fill=white](u71) at (3,9.5)[] {};
\node[draw, circle, line width=0.5pt, fill=white](u72) at (4,9.5)[] {};
\node[draw, circle, line width=0.5pt, fill=white](u73) at (5,9.5)[] {};
\node[draw, circle, line width=0.5pt, fill=white](u74) at (6,9.5)[] {};
\node[draw, circle, line width=0.5pt, fill=white](u75) at (9,9.5)[] {};
\node[draw, circle, line width=0.5pt, fill=white](u76) at (10,9.5)[] {};

\draw[-, line width=1pt, color=red]  (u) -- (u7);
\draw[-, line width=1pt, color=red]  (u7) -- (u40);
\draw[-, line width=1pt, color=red]  (u7) -- (u70);
\draw[-, line width=1pt, color=red]  (u70) -- (u51);
\draw[-, line width=1pt, color=red]  (u51) -- (u72);
\draw[-, line width=1pt, color=red]  (u72) -- (u53);
\draw[-, line width=1pt, color=red]  (u53) -- (u74);
\draw[-, line width=1pt, color=red]  (u75) -- (u76);

\draw[-, line width=0.5pt]  (u50) -- (u51);
\draw[-, line width=0.5pt]  (u51) -- (u52);
\draw[-, line width=0.5pt]  (u52) -- (u53);
\draw[-, line width=0.5pt]  (u53) -- (u54);
\draw[-, line width=0.5pt]  (u55) -- (u56);

\draw[-, line width=0.5pt]  (u70) -- (u60);
\draw[-, line width=0.5pt]  (u60) -- (u50);
\draw[-, line width=0.5pt]  (u72) -- (u62);
\draw[-, line width=0.5pt]  (u62) -- (u52);
\draw[-, line width=0.5pt]  (u74) -- (u64);
\draw[-, line width=0.5pt]  (u64) -- (u54);
\draw[-, line width=0.5pt]  (u76) -- (u66);
\draw[-, line width=0.5pt]  (u66) -- (u56);

\draw[-, line width=0.5pt]  (u70) -- (u61);
\draw[-, line width=0.5pt]  (u61) -- (u72);
\draw[-, line width=0.5pt]  (u72) -- (u63);
\draw[-, line width=0.5pt]  (u63) -- (u74);
\draw[-, line width=0.5pt]  (u55) -- (u76);
\draw[-, line width=0.5pt]  (u65) -- (u76);

\draw[-, line width=0.5pt]  (u70) -- (u71);
\draw[-, line width=0.5pt]  (u71) -- (u72);
\draw[-, line width=0.5pt]  (u72) -- (u73);
\draw[-, line width=0.5pt]  (u73) -- (u74);

\path (6,8.5) -- (9,8.5) node [black, midway, sloped] {\Large$\dots$};
\node at (12.5,   8.5) {$H^1$};
\node[left = 0.05 of u7] () {\includegraphics[scale=0.06]{robot-red.png}};
\node[left = 0.05 of u] () {\includegraphics[scale=0.06]{robot-red.png}};

\node[draw, circle, line width=0.5pt, fill=white](w6) at (11,7.5)[] {};
\node[draw, circle, line width=0.5pt, fill=white](w7) at (11,8.5)[] {};
\node[draw, circle, line width=0.5pt, fill=red](w8) at (11,9.5)[] {};

\draw[-, line width=0.5pt]  (u56) -- (w6);
\draw[-, line width=0.5pt]  (w6) -- (w7);
\draw[-, line width=0.5pt]  (w7) -- (w8);
\draw[-, line width=1pt, color=red]  (w8) -- (u76);

\end{tikzpicture}
\caption{High-level representation of the instance constructed in the proof of Theorem~\ref{thm:keeping:the:same:makespan:is:NP:hard}. The depicted red paths are given as an example on how these paths could behave.}\label{fig:hardness}
\end{figure}
We achieve this result via a reduction from \textsc{$3$-SAT}, a classic \NP-hard problem~\cite{GJ75}. Let $\phi$ be an instance of \textsc{$3$-SAT} over a set of variables $X=\{x_1,\ldots,x_n\}$ and set of clauses $C=\{c_1\ldots,c_m\}$.
We create an instance $\mathcal{I}=\langle G, A, s_0, t \rangle$ together with a schedule $\sigma$ and prove that there exists an agent $a\in A$ such that if $a$ malfunctions-$1$ in turn $1$, then we can compute a set of delay-$1$ operations that result in a feasible schedule of $\mathcal{I}$ of makespan $\ell$ if and only if $\phi$ is a yes-instance of \textsc{$3$-SAT}. 
A simplified version of the instance we construct is depicted in Figure~\ref{fig:hardness}. The constructed graph $G$ is built using $2n+m$ gadgets $H$. We stress here that the proof is highly technical, and overly long. So we first present a very high-level sketch before proceeding with the whole proof. 

\paragraph{Sketch.}
Let us analyze the behavior of each type of agent, according to their color. 
First, there are two blue agents $b^i_{yes}$ and $b^i_{no}$ per variable in $x_i\in X$; the first is intended to follow the blue path $P^i_{yes}$, while the second the blue path $P^i_{no}$. The length of these paths is set to exactly $\ell-1$. This will be important later, as it does not allow any blue agent to perform more than one delay-$1$ operation. Then there is one red agent $r_j$ for each $H^j$ red agents. The agent that corresponds to $H^j$ is the one whose starting position is \textit{above} $H^j$. So, the first red agent that appears from the top is $r^1$, while the last is $r^{2n+m}$. According to the given schedule $\sigma$, for each $j$ the agent $r^j$ is expected to move down, enter $H^j$, follow its red path and finish at the red vertex of $H^j$. The exact red path for each red agent depends on the formula $\phi$. According to $\sigma$, the red agents are able to move without stopping and without interfering with the blue agents. Finally, there is the brown agent who just moves to the brown vertex. 

The key point is that the schedule $\sigma$ hinges on the brown agent moving on the first round. If, instead, the brown agent performs a malfunction-$1$ operation on the first round, then all the red agents will be delayed by one round. This throws off the timing of $\sigma$. In particular, for any $i\in [2n+m]$, the agent $r^i$ is stuck inside the gadget $H^j$ without being able to reach its terminal, until the blue agents move thought $H^j$.

The rest of the proof depends on the following three observations. First, the red agents $r^i$, $i \in [n]$, will force at least one of the agents $b^i_{yes}, b^i_{no}$ to delay-$1$, otherwise the makespan will be more than $\ell+1$. Second, if both $b^i_{yes} $ and $ b^i_{no}$, for some $i \in [n]$, delay-$1$, then the red agent $r^{n+m+i}$ forces an extra delay-$1$ operation on one of them. This results in a makespan of at least $\ell+1$ (since no blue agent can perform more than one delay-$1$ operations).
Thus, if after the delay-$1$ operations have been performed the resulting schedule has makespan $\ell$, then we can define an assignment over the variables of $X$ by checking which blue agents have delayed-$1$.
Thirdly, the red agents $r^{n+j}$, $j \in [m]$, verify that the assignment we just defined is a satisfying $\phi$, as otherwise they introduce additional delays resulting in a schedule of makespan at least $\ell+1$. 
Therefore, given a set of delay-$1$ operations that result in a schedule of makespan $\ell$, we create a satisfying assignment of $\phi$ by identifying which blue agents perform delay-$1$ operations. 
For the reverse direction, given a satisfying assignment of $\phi$, we decide which blue agents need to perform a delay-$1$ operation. Also, we use the previous observations to prove that we can easily synchronize the blue and red agents in order to produce a schedule of makespan at most $\ell$.

\medskip 

We now formally prove all of our claims.

    \paragraph{The construction of the graph $G$}. 
    The main tool we employ is a multi-purpose gadget $H$. This gadget consists of $4n+2$ paths of length $3$. Let $\langle u_{p,1}, u_{p,2}, u_{p,3} \rangle$ be the vertices of the $p^{th}$ such path. We then add the edges:
    \begin{itemize}
        \item $u_{p,1}u_{p+1,1}$, $u_{p,3}u_{p+1,3}$, for all $p \in [4n+1]$, 
        \item $u_{2q-1,1}u_{2q,3}$, $u_{2q,3}u_{2q+1,1}$, for all $q\in [2n]$, and
        \item $u_{2q-1,1}u_{2q,2}$, $u_{2q,2}u_{2q+1,1}$, for all $q\in [2n]$.
    \end{itemize}
    Finally, for each vertex $u_{p,1}$, $p \in [4n+1]$, we crate two additional paths each of length $14n+3m+5$ and we add an edge between the fist vertex of each path and $u_{p,1}$. We will say that these are the \textit{blocking paths} and we denote them by $S_{p,1}$ and $T_{p,1}$. To be clear, the paths $S_{p,1}$ and $T_{p,1}$ are connected to the vertex $u_{p,1}$, for every $p\in[4n+1]$. For an illustration of $H$, consider the graph $H^1$ of Figure~\ref{fig:hardness}. We choose to not depict the blocking paths to avoid overloading the figure. Note that $H$ is a planar graph of maximum degree $10$, which will also be true for the final graph $G$. 

    The graph $G$ uses multiple copies of $H$ as its building blocks. In particular, we create $2n+m$ a copies of $H$, and denote by $H^i$ its $i^{th}$ copy, for $i \in [2n+m]$. We will use the superscript $j$ in order to denote a vertex that belongs to the $m$ copies of the $H$ gadget related to the cause $c_j$ of $C$. That is, we denote by $u^j_{p,q}$ the copy of $u_{p,q}$ (for $p\in[3]$ and $q\in[4n+2]$) that belongs in $H^j$. When needed, we will also denote by $S^j_{p,1}$ and $T^j_{p,1}$ the copies of the paths $S_{p,1}$ and $T_{p,1}$ respectively that appear in $H^j$. 
    We add edges between consecutive copies of gadgets $H$ as follows: for each $j \in [m-1]$ we add the edge $u^j_{p,3}u^{j+1}_{p,1}$, for all $p \in [4n+1]$. 

    Next, we create a path $P = \langle v_1,\ldots,v_{2n+m+2}\rangle$ of $2n+m+2$ new vertices (this is the path containing the red agents in Figure~\ref{fig:hardness}). We add edges in order to connect vertices of $P$ with the copies of $H$. For each $i \in [2n+m]$, we add the edge $v_{i+1}u^i_{1,1}$. Then, for each vertex $v_{i}$, $i \in [2n+m+1]$, we create two additional paths of length $14n+3m+6$ and we add an edge between the fist vertex of each path and $v_{i}$. We also call these \textit{blocking} paths and denote them with $S_{i}$ and $T_{i}$ (similarly to the blocking paths we attached to vertices of $H$). 

    Finally, for each variable $x_i$, $i \in [n]$, we introduce four new paths, two consisting of $4n+3$ vertices and two of $4n+2$ vertices. In particular, we introduce two paths $P^{i}_{yes}$ and $P^{i}_{no}$ of $4n+3$ vertices each. Then, we add an edge between the last vertex of $P^{i}_{yes}$ and the vertex $u^1_{4(i-1)+2,1}$ and an edge between the last vertex of $P^{i}_{no}$ and the vertex $u^1_{4i,1}$. This corresponds to the top part of the blue paths depicted in Figure~\ref{fig:hardness}. After that we introduce the two paths $Q^{i}_{yes}$ and $Q^{i}_{no}$ consisting of $4n+2$ vertices each. We add an edge between the first vertex of $Q^{i}_{yes}$ and the vertex $u^{2n+m}_{4(i-1)+2,3}$ and an edge between the first vertex of $Q^{i}_{no}$ and the vertex $u^{2n+m}_{4i,3}$. This corresponds to the bottom part of the blue paths depicted in Figure~\ref{fig:hardness}. This finishes the construction of $G$. 

    \paragraph{Definition of $A$, $s_0$ and $t$.} 
    We begin by introducing two agents for each variable $x_i\in X$; we call these \textit{variable agents} (blue agents in Figure~\ref{fig:hardness}). Let $\nu_i$ and $\mu_i$ be the two agents we created for $x_i$, for $i \in [n]$. We set $s_0(\nu_i)$ to be the first vertex of $P^i_{yes}$ and $t(\nu_i)$ to be the last vertex of $Q^i_{yes}$. Similarly, we set $s_0(\mu_i)$ to be the first vertex of $P^i_{no}$ and $t(\mu_i)$ to be the last vertex of $Q^i_{no}$. 
    
    We then define the red agents of Figure~\ref{fig:hardness}. There are three types of red agents. First, we create the $\lambda_i$ agents, one for each variable $i \in [n]$; we call these \textit{select-assignment agents}. We set $s_0(\lambda_i)= v_i$ and $t(\lambda_i)=u^i_{4n+2,1}$, for all $i \in [n]$. 
    We also introduce a set of \textit{verifying-assignment agents}. For each $i\in [n]$, we create an agent $\tau_i$ and set $s_0(\tau_i)= v_{n+m+i}$ and $t(\tau_i)=u^{n+m+i}_{4n+2,1}$. We proceed with the set of \textit{clause agents}. For each $j \in [m]$ we introduce a clause agent $\kappa_j$ and set $s_0(\kappa_j)= v_{n+j}$ and $t(\kappa_j)=u^{n+j}_{4n+2,1}$.
    Lastly, the malfunction will be performed by one \textit{introduce delay agent} $a$ such that $s_0(a) = v_{2n+m+1}$ and $t(a)=v_{2n+m+2}$ (the brown agent in Figure~\ref{fig:hardness}). 
    
    Finally, we add several blocking agents that start in paths $S$ and have terminals in paths $T$. 
    To simplify the definition, we will separate these into three groups.
    
    We first consider agents for the vertices that have starting positions in the paths $S^j_{4n+1,1}$. For each $j\in [2n+m]$, we create a set of $3j$ agents $b^j_{4n+1,i}$, $i \in [3j]$. We set $s_0(b^j_i)$ to be the the $4n+2+i^{th}$ vertex of $S^j_{4n+1,1}$ and $t(b^j_i)$ to be the $10n+3m+3-i^{th}$ vertex of $T^j_{4n+1,1}$. The goal of these agents is to occupy the vertex $u^j_{4n+1,1}$ between all the turns from $4n+3$ up to $4n+3j+2$. This happens because the distance they have to move is equal to the makespan we want to achieve, i.e., $14n+3m+5$. Therefore we cannot afford to delay any of them. 

    Next we create a set of agents for each $j \in [2n+m]$. 
    We first fix a $j \in [2n+m]$. 
    Then, we create a set of $4n$ agents $b^j_{p,1}$, $p \in [4n]$. We set $s_0(b^j_{p,1})$ to be the $4n + 2 + 3j^{th}$ vertex of $S^j_{p,1}$ and $t(b^j_{p,1})$ to be the $10n+3m+ 3 - 3j^{th}$ vertex of $T^j_{p,1}$. Note that each vertex $u^j_{p,1}$ will be occupied by $b^j_{p,1}$ at the turn $4n + 2 + 3j$ as otherwise $b^j_{p,1}$ will never reach its terminal in time.

    Finally, we create a set of $10n+3m$ blocking agents $b_{i,q}$, $q \in [10n+3m]$, for every $i \in [2n+m]$. We also set $s_0(b_{i,q})$ to be the $q+2^{th}$ vertex of $S_{i+1}$ and $t(b_{i,q})$ to be the $14n+3m+3-q^{th}$ vertex of $T_{i+1}$. Once again, we know that the vertex $v_{i+1}$ will be occupied by the agents $b_{i,q}$, $q \in [10n+3m]$ from the $3$-rd turn up to the $10n+3m+2^{th}$ turn. Otherwise, at least one of these agents will never reach its terminal in time. This finishes the definition of the instance $\mathcal{I}$.

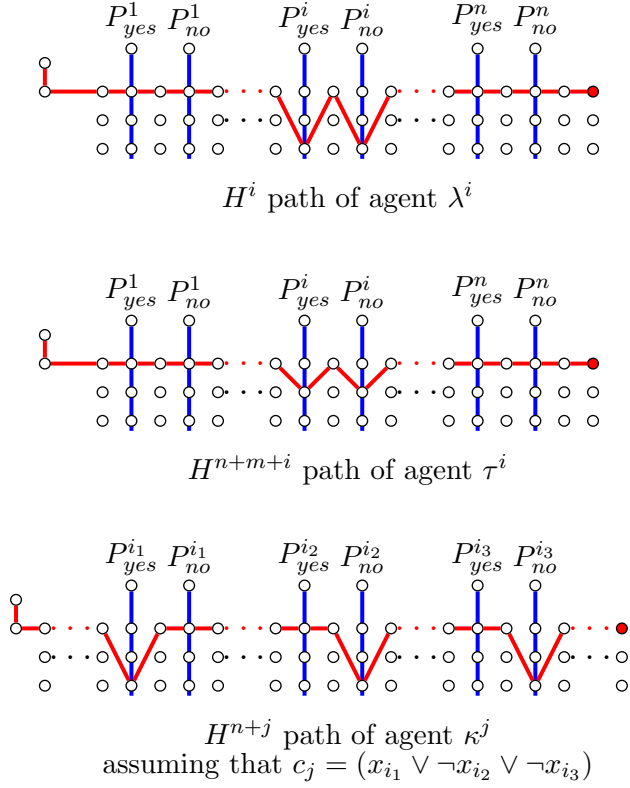
\begin{figure}[!t]
\centering
\begin{tikzpicture}[scale=0.38, inner sep=0.5mm]

\begin{scope}[yshift=-18.75cm]
\node at (3, 12) {$P^{i_1}_{yes}$};
\node at (5, 12) {$P^{i_1}_{no}$};
\node at (9, 12) {$P^{i_2}_{yes}$};
\node at (11, 12) {$P^{i_2}_{no}$};
\node at (15, 12) {$P^{i_3}_{yes}$};
\node at (17, 12) {$P^{i_3}_{no}$};

\node[draw, circle, line width=0.5pt, fill=white](b1) at (3,11)[] {};
\node[draw, circle, line width=0.5pt, fill=white](b2) at (5,11)[] {};

\node[draw, circle, line width=0.5pt, fill=white](b3) at (9,11)[] {};
\node[draw, circle, line width=0.5pt, fill=white](b4) at (11,11)[] {};

\node[draw, circle, line width=0.5pt, fill=white](b5) at (15,11)[] {};
\node[draw, circle, line width=0.5pt, fill=white](b6) at (17,11)[] {};

\node[](b7) at (3,7)[] {};
\node[](b8) at (5,7)[] {};

\node[](b9) at (9,7)[] {};
\node[](b10) at (11,7)[] {};

\node[](b11) at (15,7)[] {};
\node[](b12) at (17,7)[] {};

\draw[-, line width=1.5pt, color=blue]  (b1) -- (b7);
\draw[-, line width=1.5pt, color=blue]  (b2) -- (b8);
\draw[-, line width=1.5pt, color=blue]  (b3) -- (b9);
\draw[-, line width=1.5pt, color=blue]  (b4) -- (b10);
\draw[-, line width=1.5pt, color=blue]  (b5) -- (b11);
\draw[-, line width=1.5pt, color=blue]  (b6) -- (b12);


\node[draw, circle, line width=0.5pt, fill=white](u) at (-1,10.5)[] {};
\node[draw, circle, line width=0.5pt, fill=white](u1) at (-1,9.5)[] {};
\node[draw, circle, line width=0.5pt, fill=white](v10) at (0,9.5)[] {};
\node[draw, circle, line width=0.5pt, fill=white](v20) at (0,8.5)[] {};
\node[draw, circle, line width=0.5pt, fill=white](v30) at (0,7.5)[] {};

\node[draw, circle, line width=0.5pt, fill=white](u10) at (2,9.5)[] {};
\node[draw, circle, line width=0.5pt, fill=white](u11) at (3,9.5)[] {};
\node[draw, circle, line width=0.5pt, fill=white](u12) at (4,9.5)[] {};
\node[draw, circle, line width=0.5pt, fill=white](u13) at (5,9.5)[] {};
\node[draw, circle, line width=0.5pt, fill=white](u14) at (6,9.5)[] {};

\node[draw, circle, line width=0.5pt, fill=white](u15) at (8,9.5)[] {};
\node[draw, circle, line width=0.5pt, fill=white](u16) at (9,9.5)[] {};
\node[draw, circle, line width=0.5pt, fill=white](u17) at (10,9.5)[] {};
\node[draw, circle, line width=0.5pt, fill=white](u18) at (11,9.5)[] {};
\node[draw, circle, line width=0.5pt, fill=white](u19) at (12,9.5)[] {};

\node[draw, circle, line width=0.5pt, fill=white](u110) at (14,9.5)[] {};
\node[draw, circle, line width=0.5pt, fill=white](u111) at (15,9.5)[] {};
\node[draw, circle, line width=0.5pt, fill=white](u112) at (16,9.5)[] {};
\node[draw, circle, line width=0.5pt, fill=white](u113) at (17,9.5)[] {};
\node[draw, circle, line width=0.5pt, fill=white](u114) at (18,9.5)[] {};
\node[draw, circle, line width=0.5pt, fill=red](u115) at (20,9.5)[] {};

\node[draw, circle, line width=0.5pt, fill=white](u20) at (2,8.5)[] {};
\node[draw, circle, line width=0.5pt, fill=white](u21) at (3,8.5)[] {};
\node[draw, circle, line width=0.5pt, fill=white](u22) at (4,8.5)[] {};
\node[draw, circle, line width=0.5pt, fill=white](u23) at (5,8.5)[] {};
\node[draw, circle, line width=0.5pt, fill=white](u24) at (6,8.5)[] {};

\node[draw, circle, line width=0.5pt, fill=white](u25) at (8,8.5)[] {};
\node[draw, circle, line width=0.5pt, fill=white](u26) at (9,8.5)[] {};
\node[draw, circle, line width=0.5pt, fill=white](u27) at (10,8.5)[] {};
\node[draw, circle, line width=0.5pt, fill=white](u28) at (11,8.5)[] {};
\node[draw, circle, line width=0.5pt, fill=white](u29) at (12,8.5)[] {};

\node[draw, circle, line width=0.5pt, fill=white](u210) at (14,8.5)[] {};
\node[draw, circle, line width=0.5pt, fill=white](u211) at (15,8.5)[] {};
\node[draw, circle, line width=0.5pt, fill=white](u212) at (16,8.5)[] {};
\node[draw, circle, line width=0.5pt, fill=white](u213) at (17,8.5)[] {};
\node[draw, circle, line width=0.5pt, fill=white](u214) at (18,8.5)[] {};
\node[draw, circle, line width=0.5pt, fill=white](u215) at (20,8.5)[] {};

\node[draw, circle, line width=0.5pt, fill=white](u30) at (2,7.5)[] {};
\node[draw, circle, line width=0.5pt, fill=white](u31) at (3,7.5)[] {};
\node[draw, circle, line width=0.5pt, fill=white](u32) at (4,7.5)[] {};
\node[draw, circle, line width=0.5pt, fill=white](u33) at (5,7.5)[] {};
\node[draw, circle, line width=0.5pt, fill=white](u34) at (6,7.5)[] {};

\node[draw, circle, line width=0.5pt, fill=white](u35) at (8,7.5)[] {};
\node[draw, circle, line width=0.5pt, fill=white](u36) at (9,7.5)[] {};
\node[draw, circle, line width=0.5pt, fill=white](u37) at (10,7.5)[] {};
\node[draw, circle, line width=0.5pt, fill=white](u38) at (11,7.5)[] {};
\node[draw, circle, line width=0.5pt, fill=white](u39) at (12,7.5)[] {};

\node[draw, circle, line width=0.5pt, fill=white](u310) at (14,7.5)[] {};
\node[draw, circle, line width=0.5pt, fill=white](u311) at (15,7.5)[] {};
\node[draw, circle, line width=0.5pt, fill=white](u312) at (16,7.5)[] {};
\node[draw, circle, line width=0.5pt, fill=white](u313) at (17,7.5)[] {};
\node[draw, circle, line width=0.5pt, fill=white](u314) at (18,7.5)[] {};
\node[draw, circle, line width=0.5pt, fill=white](u315) at (20,7.5)[] {};

\draw[-, line width=1.5pt, color=red]  (u) -- (u1);
\draw[-, line width=1.5pt, color=red]  (u1) -- (v10);
\draw[-, line width=1.5pt, color=red]  (u10) -- (u31);
\draw[-, line width=1.5pt, color=red]  (u31) -- (u12);
\draw[-, line width=1.5pt, color=red]  (u16) -- (u17);
\draw[-, line width=1.5pt, color=red]  (u15) -- (u16);
\draw[-, line width=1.5pt, color=red]  (u17) -- (u38);
\draw[-, line width=1.5pt, color=red]  (u38) -- (u19);
\draw[-, line width=1.5pt, color=red]  (u110) -- (u111);
\draw[-, line width=1.5pt, color=red]  (u111) -- (u112);
\draw[-, line width=1.5pt, color=red]  (u112) -- (u313);
\draw[-, line width=1.5pt, color=red]  (u313) -- (u114);

\draw[-, line width=1.5pt,  color=red]  (u12) -- (u13);
\draw[-, line width=1.5pt,  color=red]  (u13) -- (u14);

\path (0,8.5) -- (2,8.5) node [black, midway] {\Large{$\dots$}};
\path (6,8.5) -- (8,8.5) node [black, midway, sloped] {\Large$\dots$};
\path (12,8.5) -- (14,8.5) node [black, midway, sloped] {\Large$\dots$};
\path (18,8.5) -- (20,8.5) node [black, midway] {\Large{$\dots$}};

\path (0,9.5) -- (2,9.5) node [red, midway] {\Large{$\dots$}};
\path (6,9.5) -- (8,9.5) node [red, midway] {\Large{$\dots$}};
\path (12,9.5) -- (14,9.5) node [red, midway, sloped] {\Large$\dots$};
\path (18,9.5) -- (20,9.5) node [red, midway] {\Large{$\dots$}};

\node at (10.5,   5.8) {$H^{n+j}$ path of agent $\kappa^j$};
\node at (10.5,   4.7) {assuming that $c_j = (x_{i_1}\lor \neg x_{i_2}\lor \neg x_{i_3} )$};

\end{scope}

\begin{scope}
\node at (3, 12) {$P^1_{yes}$};
\node at (5, 12) {$P^1_{no}$};
\node at (9, 12) {$P^i_{yes}$};
\node at (11, 12) {$P^i_{no}$};
\node at (15, 12) {$P^n_{yes}$};
\node at (17, 12) {$P^n_{no}$};

\node[draw, circle, line width=0.5pt, fill=white](b1) at (3,11)[] {};
\node[draw, circle, line width=0.5pt, fill=white](b2) at (5,11)[] {};

\node[draw, circle, line width=0.5pt, fill=white](b3) at (9,11)[] {};
\node[draw, circle, line width=0.5pt, fill=white](b4) at (11,11)[] {};

\node[draw, circle, line width=0.5pt, fill=white](b5) at (15,11)[] {};
\node[draw, circle, line width=0.5pt, fill=white](b6) at (17,11)[] {};

\node[](b7) at (3,7)[] {};
\node[](b8) at (5,7)[] {};

\node[](b9) at (9,7)[] {};
\node[](b10) at (11,7)[] {};

\node[](b11) at (15,7)[] {};
\node[](b12) at (17,7)[] {};

\draw[-, line width=1.5pt, color=blue]  (b1) -- (b7);
\draw[-, line width=1.5pt, color=blue]  (b2) -- (b8);
\draw[-, line width=1.5pt, color=blue]  (b3) -- (b9);
\draw[-, line width=1.5pt, color=blue]  (b4) -- (b10);
\draw[-, line width=1.5pt, color=blue]  (b5) -- (b11);
\draw[-, line width=1.5pt, color=blue]  (b6) -- (b12);


\node[draw, circle, line width=0.5pt, fill=white](u) at (0,10.5)[] {};
\node[draw, circle, line width=0.5pt, fill=white](u1) at (0,9.5)[] {};
\node[draw, circle, line width=0.5pt, fill=white](u10) at (2,9.5)[] {};
\node[draw, circle, line width=0.5pt, fill=white](u11) at (3,9.5)[] {};
\node[draw, circle, line width=0.5pt, fill=white](u12) at (4,9.5)[] {};
\node[draw, circle, line width=0.5pt, fill=white](u13) at (5,9.5)[] {};
\node[draw, circle, line width=0.5pt, fill=white](u14) at (6,9.5)[] {};

\node[draw, circle, line width=0.5pt, fill=white](u15) at (8,9.5)[] {};
\node[draw, circle, line width=0.5pt, fill=white](u16) at (9,9.5)[] {};
\node[draw, circle, line width=0.5pt, fill=white](u17) at (10,9.5)[] {};
\node[draw, circle, line width=0.5pt, fill=white](u18) at (11,9.5)[] {};
\node[draw, circle, line width=0.5pt, fill=white](u19) at (12,9.5)[] {};

\node[draw, circle, line width=0.5pt, fill=white](u110) at (14,9.5)[] {};
\node[draw, circle, line width=0.5pt, fill=white](u111) at (15,9.5)[] {};
\node[draw, circle, line width=0.5pt, fill=white](u112) at (16,9.5)[] {};
\node[draw, circle, line width=0.5pt, fill=white](u113) at (17,9.5)[] {};
\node[draw, circle, line width=0.5pt, fill=white](u114) at (18,9.5)[] {};
\node[draw, circle, line width=0.5pt, fill=red](u115) at (19,9.5)[] {};

\node[draw, circle, line width=0.5pt, fill=white](u20) at (2,8.5)[] {};
\node[draw, circle, line width=0.5pt, fill=white](u21) at (3,8.5)[] {};
\node[draw, circle, line width=0.5pt, fill=white](u22) at (4,8.5)[] {};
\node[draw, circle, line width=0.5pt, fill=white](u23) at (5,8.5)[] {};
\node[draw, circle, line width=0.5pt, fill=white](u24) at (6,8.5)[] {};

\node[draw, circle, line width=0.5pt, fill=white](u25) at (8,8.5)[] {};
\node[draw, circle, line width=0.5pt, fill=white](u26) at (9,8.5)[] {};
\node[draw, circle, line width=0.5pt, fill=white](u27) at (10,8.5)[] {};
\node[draw, circle, line width=0.5pt, fill=white](u28) at (11,8.5)[] {};
\node[draw, circle, line width=0.5pt, fill=white](u29) at (12,8.5)[] {};

\node[draw, circle, line width=0.5pt, fill=white](u210) at (14,8.5)[] {};
\node[draw, circle, line width=0.5pt, fill=white](u211) at (15,8.5)[] {};
\node[draw, circle, line width=0.5pt, fill=white](u212) at (16,8.5)[] {};
\node[draw, circle, line width=0.5pt, fill=white](u213) at (17,8.5)[] {};
\node[draw, circle, line width=0.5pt, fill=white](u214) at (18,8.5)[] {};
\node[draw, circle, line width=0.5pt, fill=white](u215) at (19,8.5)[] {};

\node[draw, circle, line width=0.5pt, fill=white](u30) at (2,7.5)[] {};
\node[draw, circle, line width=0.5pt, fill=white](u31) at (3,7.5)[] {};
\node[draw, circle, line width=0.5pt, fill=white](u32) at (4,7.5)[] {};
\node[draw, circle, line width=0.5pt, fill=white](u33) at (5,7.5)[] {};
\node[draw, circle, line width=0.5pt, fill=white](u34) at (6,7.5)[] {};

\node[draw, circle, line width=0.5pt, fill=white](u35) at (8,7.5)[] {};
\node[draw, circle, line width=0.5pt, fill=white](u36) at (9,7.5)[] {};
\node[draw, circle, line width=0.5pt, fill=white](u37) at (10,7.5)[] {};
\node[draw, circle, line width=0.5pt, fill=white](u38) at (11,7.5)[] {};
\node[draw, circle, line width=0.5pt, fill=white](u39) at (12,7.5)[] {};

\node[draw, circle, line width=0.5pt, fill=white](u310) at (14,7.5)[] {};
\node[draw, circle, line width=0.5pt, fill=white](u311) at (15,7.5)[] {};
\node[draw, circle, line width=0.5pt, fill=white](u312) at (16,7.5)[] {};
\node[draw, circle, line width=0.5pt, fill=white](u313) at (17,7.5)[] {};
\node[draw, circle, line width=0.5pt, fill=white](u314) at (18,7.5)[] {};
\node[draw, circle, line width=0.5pt, fill=white](u315) at (19,7.5)[] {};

\draw[-, line width=1.5pt, color=red]  (u) -- (u1);
\draw[-, line width=1.5pt, color=red]  (u1) -- (u10);
\draw[-, line width=1.5pt, color=red]  (u10) -- (u11);
\draw[-, line width=1.5pt, color=red]  (u11) -- (u12);
\draw[-, line width=1.5pt, color=red]  (u15) -- (u36);
\draw[-, line width=1.5pt, color=red]  (u36) -- (u17);
\draw[-, line width=1.5pt, color=red]  (u17) -- (u38);
\draw[-, line width=1.5pt, color=red]  (u38) -- (u19);
\draw[-, line width=1.5pt, color=red]  (u110) -- (u111);
\draw[-, line width=1.5pt, color=red]  (u111) -- (u112);
\draw[-, line width=1.5pt, color=red]  (u112) -- (u113);
\draw[-, line width=1.5pt, color=red]  (u113) -- (u114);
\draw[-, line width=1.5pt, color=red]  (u114) -- (u115);

\draw[-, line width=1.5pt,  color=red]  (u12) -- (u13);
\draw[-, line width=1.5pt,  color=red]  (u13) -- (u14);

\path (6,8.5) -- (8,8.5) node [black, midway, sloped] {\Large$\dots$};
\path (12,8.5) -- (14,8.5) node [black, midway, sloped] {\Large$\dots$};
\path (6,9.5) -- (8,9.5) node [red, midway] {\Large{$\dots$}};
\path (12,9.5) -- (14,9.5) node [red, midway, sloped] {\Large$\dots$};
\node at (10.5,   5.8) {$H^{i}$ path of agent $\lambda^i$};

\end{scope}

\begin{scope}[yshift=-9.5cm]
\node at (3, 12) {$P^1_{yes}$};
\node at (5, 12) {$P^1_{no}$};
\node at (9, 12) {$P^i_{yes}$};
\node at (11, 12) {$P^i_{no}$};
\node at (15, 12) {$P^n_{yes}$};
\node at (17, 12) {$P^n_{no}$};

\node[draw, circle, line width=0.5pt, fill=white](b1) at (3,11)[] {};
\node[draw, circle, line width=0.5pt, fill=white](b2) at (5,11)[] {};

\node[draw, circle, line width=0.5pt, fill=white](b3) at (9,11)[] {};
\node[draw, circle, line width=0.5pt, fill=white](b4) at (11,11)[] {};

\node[draw, circle, line width=0.5pt, fill=white](b5) at (15,11)[] {};
\node[draw, circle, line width=0.5pt, fill=white](b6) at (17,11)[] {};

\node[](b7) at (3,7)[] {};
\node[](b8) at (5,7)[] {};

\node[](b9) at (9,7)[] {};
\node[](b10) at (11,7)[] {};

\node[](b11) at (15,7)[] {};
\node[](b12) at (17,7)[] {};

\draw[-, line width=1.5pt, color=blue]  (b1) -- (b7);
\draw[-, line width=1.5pt, color=blue]  (b2) -- (b8);
\draw[-, line width=1.5pt, color=blue]  (b3) -- (b9);
\draw[-, line width=1.5pt, color=blue]  (b4) -- (b10);
\draw[-, line width=1.5pt, color=blue]  (b5) -- (b11);
\draw[-, line width=1.5pt, color=blue]  (b6) -- (b12);


\node[draw, circle, line width=0.5pt, fill=white](u) at (0,10.5)[] {};
\node[draw, circle, line width=0.5pt, fill=white](u1) at (0,9.5)[] {};
\node[draw, circle, line width=0.5pt, fill=white](u10) at (2,9.5)[] {};
\node[draw, circle, line width=0.5pt, fill=white](u11) at (3,9.5)[] {};
\node[draw, circle, line width=0.5pt, fill=white](u12) at (4,9.5)[] {};
\node[draw, circle, line width=0.5pt, fill=white](u13) at (5,9.5)[] {};
\node[draw, circle, line width=0.5pt, fill=white](u14) at (6,9.5)[] {};

\node[draw, circle, line width=0.5pt, fill=white](u15) at (8,9.5)[] {};
\node[draw, circle, line width=0.5pt, fill=white](u16) at (9,9.5)[] {};
\node[draw, circle, line width=0.5pt, fill=white](u17) at (10,9.5)[] {};
\node[draw, circle, line width=0.5pt, fill=white](u18) at (11,9.5)[] {};
\node[draw, circle, line width=0.5pt, fill=white](u19) at (12,9.5)[] {};

\node[draw, circle, line width=0.5pt, fill=white](u110) at (14,9.5)[] {};
\node[draw, circle, line width=0.5pt, fill=white](u111) at (15,9.5)[] {};
\node[draw, circle, line width=0.5pt, fill=white](u112) at (16,9.5)[] {};
\node[draw, circle, line width=0.5pt, fill=white](u113) at (17,9.5)[] {};
\node[draw, circle, line width=0.5pt, fill=white](u114) at (18,9.5)[] {};
\node[draw, circle, line width=0.5pt, fill=red](u115) at (19,9.5)[] {};

\node[draw, circle, line width=0.5pt, fill=white](u20) at (2,8.5)[] {};
\node[draw, circle, line width=0.5pt, fill=white](u21) at (3,8.5)[] {};
\node[draw, circle, line width=0.5pt, fill=white](u22) at (4,8.5)[] {};
\node[draw, circle, line width=0.5pt, fill=white](u23) at (5,8.5)[] {};
\node[draw, circle, line width=0.5pt, fill=white](u24) at (6,8.5)[] {};

\node[draw, circle, line width=0.5pt, fill=white](u25) at (8,8.5)[] {};
\node[draw, circle, line width=0.5pt, fill=white](u26) at (9,8.5)[] {};
\node[draw, circle, line width=0.5pt, fill=white](u27) at (10,8.5)[] {};
\node[draw, circle, line width=0.5pt, fill=white](u28) at (11,8.5)[] {};
\node[draw, circle, line width=0.5pt, fill=white](u29) at (12,8.5)[] {};

\node[draw, circle, line width=0.5pt, fill=white](u210) at (14,8.5)[] {};
\node[draw, circle, line width=0.5pt, fill=white](u211) at (15,8.5)[] {};
\node[draw, circle, line width=0.5pt, fill=white](u212) at (16,8.5)[] {};
\node[draw, circle, line width=0.5pt, fill=white](u213) at (17,8.5)[] {};
\node[draw, circle, line width=0.5pt, fill=white](u214) at (18,8.5)[] {};
\node[draw, circle, line width=0.5pt, fill=white](u215) at (19,8.5)[] {};

\node[draw, circle, line width=0.5pt, fill=white](u30) at (2,7.5)[] {};
\node[draw, circle, line width=0.5pt, fill=white](u31) at (3,7.5)[] {};
\node[draw, circle, line width=0.5pt, fill=white](u32) at (4,7.5)[] {};
\node[draw, circle, line width=0.5pt, fill=white](u33) at (5,7.5)[] {};
\node[draw, circle, line width=0.5pt, fill=white](u34) at (6,7.5)[] {};

\node[draw, circle, line width=0.5pt, fill=white](u35) at (8,7.5)[] {};
\node[draw, circle, line width=0.5pt, fill=white](u36) at (9,7.5)[] {};
\node[draw, circle, line width=0.5pt, fill=white](u37) at (10,7.5)[] {};
\node[draw, circle, line width=0.5pt, fill=white](u38) at (11,7.5)[] {};
\node[draw, circle, line width=0.5pt, fill=white](u39) at (12,7.5)[] {};

\node[draw, circle, line width=0.5pt, fill=white](u310) at (14,7.5)[] {};
\node[draw, circle, line width=0.5pt, fill=white](u311) at (15,7.5)[] {};
\node[draw, circle, line width=0.5pt, fill=white](u312) at (16,7.5)[] {};
\node[draw, circle, line width=0.5pt, fill=white](u313) at (17,7.5)[] {};
\node[draw, circle, line width=0.5pt, fill=white](u314) at (18,7.5)[] {};
\node[draw, circle, line width=0.5pt, fill=white](u315) at (19,7.5)[] {};

\draw[-, line width=1.5pt, color=red]  (u) -- (u1);
\draw[-, line width=1.5pt, color=red]  (u1) -- (u10);
\draw[-, line width=1.5pt, color=red]  (u10) -- (u11);
\draw[-, line width=1.5pt, color=red]  (u11) -- (u12);
\draw[-, line width=1.5pt, color=red]  (u15) -- (u26);
\draw[-, line width=1.5pt, color=red]  (u26) -- (u17);
\draw[-, line width=1.5pt, color=red]  (u17) -- (u28);
\draw[-, line width=1.5pt, color=red]  (u28) -- (u19);
\draw[-, line width=1.5pt, color=red]  (u110) -- (u111);
\draw[-, line width=1.5pt, color=red]  (u111) -- (u112);
\draw[-, line width=1.5pt, color=red]  (u112) -- (u113);
\draw[-, line width=1.5pt, color=red]  (u113) -- (u114);
\draw[-, line width=1.5pt, color=red]  (u114) -- (u115);

\draw[-, line width=1.5pt,  color=red]  (u12) -- (u13);
\draw[-, line width=1.5pt,  color=red]  (u13) -- (u14);

\path (6,8.5) -- (8,8.5) node [black, midway, sloped] {\Large$\dots$};
\path (12,8.5) -- (14,8.5) node [black, midway, sloped] {\Large$\dots$};
\path (6,9.5) -- (8,9.5) node [red, midway] {\Large{$\dots$}};
\path (12,9.5) -- (14,9.5) node [red, midway, sloped] {\Large$\dots$};
\node at (10.5, 5.8) {$H^{n+m+i}$ path of agent $\tau^i$};

\end{scope}

\end{tikzpicture}
\caption{A depiction of the paths followed by the select-assignment, verifying-assignment and clause agents $\lambda^i, \tau^i$ and $k^j$ respectively, according to the schedule $\sigma$ described in the proof of Theorem~\ref{thm:keeping:the:same:makespan:is:NP:hard}.}\label{fig:hardness2}
\end{figure}

    
    \paragraph{Definition of $\sigma$.}For any blocking agent, we define their schedule as the shortest path between their starting and terminal positions. Note that these paths are of length $14n+3m+5$. 

    Next, we give the schedule for the variable agents $\nu_i$ and $\mu_i$, $i \in [n]$. We set these agents to move thought the shortest paths between their starting and terminal positions, without stopping, until they reach their destinations. In particular,
    \begin{itemize}
        \item The path of $\nu_i$ consists of $P^i_{yes}, u^1_{4(i-1)+2,1},u^1_{4(i-1)+2,2},u^1_{4(i-1)+2,3},\ldots, u^{2n+m}_{4(i-1)+2,1},$ \\ $u^{2n+m}_{4(i-1)+2,2},u^{2n+m}_{4(i-1)+2,3}, Q^i_{yes},t(\nu_i)$, where the vertices of $P^i_{yes}$, $Q^i_{yes}$ are assumed to appear in the appropriate order and $t(\nu_i)$ is repeated twice at the end (as it already appears in $Q^i_{yes}$).
        \item The path of $\mu_i$ consists of $P^i_{no}, u^1_{4i,1},u^1_{4i,2},u^1_{4i,3},\ldots, u^{2n+m}_{4i,1},u^{2n+m}_{4i,2},u^{2n+m}_{4i,3}$ $, Q^i_{no},t(\mu_i)$, where the vertices of $P^i_{no}$, $Q^i_{no}$ are assumed to appear in the appropriate order and $t(\mu_i)$ is repeated twice in the end (as it already appears in $Q^i_{no}$).  
    \end{itemize}

    The schedule of the select-assignment agents $\lambda_i$, $i \in [n]$ is as follows (see also the first graph in Figure~\ref{fig:hardness2}). The agent $\lambda_i$ reaches its terminal in the turn $4n+3$ and stays there for the rest of the schedule. Therefore, we will just give their schedule for the first $4n+3$ turns. We set: 
    \begin{itemize}
        \item $s_1(\lambda_i)=v_{i+1}$,
        \item $s_{l}(\lambda_i)= u^i_{l-1,1}$, for any turn $l \in [4n+3]\setminus \{1, 4(i-1)+3, 4i+1\}$, and
        \item $s_l(\lambda_i) = u^i_{l-1,3}$, for $l\in \{4(i-1)+3, 4i+1\}$.
    \end{itemize}

    Then we have the schedules of the verifying-assignment agents $\tau_i$, $i \in [n]$ (see also the second graph in Figure~\ref{fig:hardness2}). The agent $\tau_i$ reaches its terminal in the turn $4n+3$ and stays there for the rest of the schedule. Therefore, we will just give their schedule for the first $4n+3$ turns. 
    We set: 
    \begin{itemize}
        \item $s_1(\tau_i)=v_{n+m+i+1}$,
        \item $s_{l}(\tau_i)= u^{n+m+i}_{l-1,1}$, for any turn $l \in [4n+3]\setminus \{1, 4(i-1)+3, 4i+1\}$, and
        \item $s_l(\tau_i) = u^{n+m+i}_{l-1,2}$, for $l\in \{4(i-1)+3, 4i+1\}$.
    \end{itemize}

    We proceed with the definition of the schedule for the clause agents $\kappa_j$, $j \in [m]$ (see also the third graph in Figure~\ref{fig:hardness2}). In order to define this schedule we first define a set of indices $I_j$ and a set of vertices $V_j$. To define $I_j$, we consider the literals that appear in the clause $c_j$. In particular, let $\{i_1,i_2,i_3\}\subseteq [n]$ be the set of indices such that, for all $i \in \{i_1,i_2,i_3\}$, the variable $x_i$ appear in $c_j$. For any $i \in \{i_1,i_2,i_3\}$, if $x_i$ appear in $c_j$ as a positive literal we include in $I_j$ the index $4(i-1)+3$, otherwise we include in $I_j$ the index $4i+1$. Then we set $V_j = \{u^{n+j}_{i-1,3} \mid i \in I_j\}$.
    We are ready to define the schedule for $\kappa_j$. We set: 
    \begin{itemize}
        \item $s_1(\kappa_j)=v_{n+j+1}$,
        \item $s_{l}(\kappa_j)= u^{n+j}_{l-1,1}$, for any turn $l \in [4n+3]\setminus (I_j \cup  \{1\})$, and
        \item $s_l(\kappa_j) = u^{n+j}_{l-1,3}$, for $l\in I_j$.
    \end{itemize}
    Similarly to the $\lambda$ and $\tau$ agents, each $\kappa_j$ agent reaches its terminal in turn $4n+3$ and stays there for the rest of the schedule.
    
    Finally, we give the schedule of $a$. This agents is supposed to move to its destination during the first turn and stay there. That is, $a$ follows the schedule $s_0(a),t(a), \ldots,t(a)$.
    
    This completes the construction of the instance.

    \paragraph{The reduction.}
    We will now prove that if the agent $a$ performs a malfunction-$1$ operation in the first turn, then there exists a sequence of delay-$1$ operations that results in a schedule of length $14n+3m+5$ for $\mathcal{I}$ if and only if the given formula $\phi$ of \textsc{$3$-SAT} is satisfiable.

    We first show that given a satisfying assignment of $\phi$ we can compute the sequence of delay-$1$ operations that results to a schedule of makespan $14n+3m+5$.
    In particular, for each variable $x_i$, if $x_i=\textrm{true}$ we have $\nu_i$ perform a delay-$1$ operation in the first turn. Otherwise, we have agent $\mu_1$ perform a delay-$1$ operation in the first turn. Then, we deal with the red agents $\lambda$, $\tau$ and $\kappa$. 

    For each agent $\lambda_i$, we have that at least one of the vertices $u^i_{4(i-1)+2,3}$ and $u^i_{4i,3}$ is not occupied in turn $4n+2+3i$. This happens because we have introduced a delay-$1$ in one of the $\nu_i$ and $\mu_i$. Therefore, we can move the agent $\lambda_i$ until the non-occupied vertex (either $u^i_{4(i-1)+2,3}$ or $u^i_{4i,3}$) and then have them perform delay-$1$ operations up to turn $4n+2+3i$. Then, the agent moves normally until it reaches its terminal. 
    This is a valid schedule as none of the vertices $\lambda_i$ needs to pass thought is occupied from turn $2$ until turn $4n+3i$ and it can reach any of them during this time. The same holds from the turn $4n+2+3i+1$ and onward. Indeed, the only potential collision after turn $4n+2+3i$ is with agents $\nu$ and $\mu$, but all of these have left the $H^i$ gadget before turn $4n+2+3i+2$, i.e., the next time $\lambda_i$ may cross them. 

    Similar, for each agent $\tau_i$, one of the $u^{n+m+i}_{4(i-1)+2,2}$ and $u^{n+m+i}_{4i,2}$ is not occupied in turn $4n+2+3(n+m+i)$. This happens because we have introduced a delay-$1$ in exactly one of the $\nu_i$ and $\mu_i$. Therefore, in turn $4n+2+3(n+m+i)$, either $\nu_i$ has already reached $u^{n+m+i}_{4(i-1)+2,3}$ or $\mu_i$ has already reached $u^{n+m+i}_{4i,3}$.
    We will move the agent $\tau_i$ until the vertex $u^{n+m+i}_{4(i-1)+1,1}$ or $u^{n+m+i}_{4i-1,1}$ (depending on the non-delayed agent), perform delay-$1$ operations until the turn $4n+2+3(n+m+i)-1$ and then follow the original schedule without any delays. The agent $\tau_i$ can reach any of the wanted vertices before any of them is occupied. Then, after turn $4n+2+3(n+m+i)-1$, all the vertices it moves through are not occupied because no blocking agent appears in its path after turn $4n+2+3(n+m+i)$ and all $\nu$ and $\mu$ agents has left these vertices by turn $4n+2+3(n+m+i)+1$ (even if they have been delayed by one turn). 

    Now we deal with the agents $\kappa_j$. Notice that at least one of the vertices $v\in V_j$ is not occupied by an agent in turn $4n+2+3(n+j) $. This holds because the vertices in $V_j$ represent the literals that appear in $c_j$ and, since we delay agents based on the satisfying assignment, at least one of the agents that should pass from the these vertices has been delayed. 
    Therefore, we can move the agent $\kappa_j$ until they reach a vertex $v\in V_j$ that will not be occupied in the turn $4n+2+3(n+j)$, keep them there until turn $4n+2+3(n+j)$ and then follow their original schedule. 
    This schedule can be shown to be feasible following the same arguments as the agents $\lambda_i$, $i \in [n]$.

    We do not introduce any other delays. 

    Finally, all agents $\lambda$, $\tau$ and $\kappa$ reach their terminals before the turn $14n+3m+5$. Indeed, the agents that move thought the vertices of $H^j$ will start moving consecutively (until they reach their terminals) from the turn $4n+2+3j$. Also, even if any such agent was in $u^j_{1,1}$ in turn $4n+2+3j$, it reaches $u^j_{4n+2,1}$ before turn $8n+ 3 +3 j$, where $j\le 2n+m$. 
    Additionally, we have introduced at most one delay-$1$ by the agents $\nu_i$ and $\mu_i$, $i \in [n]$. Since they reach their terminals in turn $14n+3m+4$ according to $\sigma$, the still reach their terminals by turn $14n+3m+5$ according to the updated schedule. Thus, the updated schedule has makespan $14n+3m+5$.

    For the reverse direction, we show that given a set of delay-$1$ operations that result to a schedule of makespan $\ell=14n+3m+5$, we can compute a satisfying assignment of $\phi$.

    Observe first that no blocking agent can be delayed as they start at distance exactly $\ell$ from their targets. Therefore:
    \begin{itemize}
        \item any agent $\lambda^i$, $i \in [n]$, is located between $u^i_{1,1}$ and $u^i_{4n+1,1}$ from the turn $3$ until the turn $4n+2+3i$,
        \item any agent $\kappa^j$, $j\in [m]$, is located between $u^{n+j}_{1,1}$ and $u^{n+j}_{4n+1,1}$ from the turn $3$ until the turn $4n+2+3(n+j)$,
        \item any agent $\tau^i$, $i \in [n]$, is located between $u^{n+m+i}_{1,1}$ and $u^{n+m+i}_{4n+1,1}$ from the turn $3$ until the turn $4n+2+3(n+m+i)$.
    \end{itemize}

    Since the blocking agents cannot be delayed, none of the agents $\lambda^i$, $\tau^i$ and $\kappa^j$, $i \in [n]$ and $j \in [m]$, can occupy copies of $u_{p,1}$, $p \in [4n+1]$, during the turns when these are occupied by the blocking agents. 
    
    We will now prove some properties that must hold in order for the updated schedule to be of makespan $14n+3m+5$.

    \begin{claim} \label{claim:lambda:introduce:delay:to:literals}
        For each $i\in [n]$, at least one of the $\nu_i$ and $\mu_i$ must delay-$1$.
    \end{claim}

    \begin{proofclaim}
    Fix an $i \in [n]$ and consider the agent $\lambda_i$. Since we cannot delay any blocking agent, $\lambda_i$ occupies either $u^i_{4(i-1)+2,3}$ or $u^i_{4i,3}$ in turn $4n+2+3i$. Since this is the first turn that $\nu_i$ and $\mu_i$ can reach these vertices, we have that at least one of them must be delayed-$1$.
    \end{proofclaim}

    \begin{claim}\label{claim:tau:verify:that:we:have:assignment}
        For each $i\in [n]$, at most one of the $\nu_i$ and $\mu_i$ can delay-$1$.
    \end{claim}

    \begin{proofclaim}
        Fix an $i \in [n]$ and consider the agent $\tau_i$. Since we cannot delay any blocking agent, $\tau_i$ occupies either $u^{n+m+i}_{4(i-1)+2,2}$ or $u^{n+m+i}_{4i,2}$ in turn $4n+2+3(n+m+i)$. 
        This is the latest that any of the agents $\nu_i$ and $\mu_i$ can reach these vertices as, otherwise, they would have performed more delay-$1$ operations, and they wouldn't be able to  reach their terminals by turn $14n+3m+5$. Therefore, at least one of them has not been delayed.
    \end{proofclaim}

    We are now able to define an assignment over the variables $x_i$ based on which variable agents have been delayed.
    For each $i \in [n]$, if $\nu_i$ has been delayed then we set $x_i=\textrm{true}$, otherwise ($\mu_i$ has been delayed) we set $x_i=\textrm{false}$. It follows from Claims~\ref{claim:lambda:introduce:delay:to:literals} and ~\ref{claim:tau:verify:that:we:have:assignment} that:
    \begin{itemize}
        \item $x_i=\textrm{true}$ if and only if $\nu_i$ has been delayed by one turn, and,
        \item $x_i=\textrm{false}$ if and only if $\mu_i$ has been delayed by one turn. 
    \end{itemize}

    We will now use the agents $\kappa^j$, $j\in [m]$, to prove that this is indeed a satisfying assignment of the $\phi$.
    Lets fix a clause $c_j$ in $\phi$. Note that the agent $\kappa_j$ must occupy one of the vertices $v \in V_j$ in turn $4n+2+3(n+j)$, as otherwise it would delay a blocking agent. Consider now the agent $\alpha$ that passes through the vertex $v$. The agent $\alpha$ is either $\nu_i$ or $\mu_i$, for some $i\in [n]$. Also, $\alpha$ passes from $v$ on turn $4n+2+3(n+j)+1$, i.e., it is already delayed. Indeed, if this were not the case, then we would need to delay $\alpha$, which would contradict Claims~\ref{claim:lambda:introduce:delay:to:literals} and ~\ref{claim:tau:verify:that:we:have:assignment}. 
    Note that if $\alpha$ is $\nu_i$, then the literal $x_i$ appears in $c_j$. Also, the assignment we defined sets $x_i = \textrm{true}$ (as $\nu_i$ has been delayed). Similarly, if $\alpha$ is $\mu_i$, then the literal $\neg x_i$ appears in $c_j$. Also, the assignment we created sets $x_i = \textrm{false}$ (as $\mu_i$ has been delayed). 
    In both cases $c_j$ is satisfied by the truth assignment we defined.

    Therefore any set of delays resulting to a schedule of makespan $14n+3m+5$ corresponds to a satisfying assignment of $\phi$, completing the proof of the theorem.
\end{proof}

\fi

\section{Communication Protocols}
In what follows, we consider communication mechanisms that will allow the agents to take decisions by themselves. In particular, given a schedule of makespan $\ell$ we will define two protocols such that, if we have $k$ faulty agents, the agents produce a schedule of makespan at most $\ell +k$. 
    
Hereafter, we assume that, at the start of each turn any agent can decide to delay-$1$. 
This decision must happen simultaneously by all agents at the start of a turn, before any agent starts moving. In some sense, we break up a turn into two parts: the \textit{decision} and the \textit{action} phases. During the decision phase of a turn, the agents decide whether they perform a delay-$1$ operation, and during the action phase they move according to the updated schedule. 
    

\subsection{Check Before Moving}
In this section we propose our first distributed protocol for updating a schedule after a malfunction has occurred. 

\paragraph{The \protocolOneShort protocol.} Let $\mathcal{I}=\langle G, A, s_0, t, \ell\rangle$ be an instance of \MAPFShort{} and $\sigma$ be a feasible schedule of $\mathcal{I}$. 
We start with the \protocolOne (\protocolOneShort) protocol, which addresses the case where one agent malfunctions once. Before we do that, we need to introduce the following notion. We say that an agent $a\in A$ (who follows $\sigma)$ considers a vertex $v\in V$ as \textit{healthy during a turn $i$} if:

\begin{itemize}
    \item there exists an agent $b\in A$ such that $s_{i}(b)= v$, $s_{i+1}(b)\neq v$ and $b$ did not decide to perform a delay-$1$ operation during the decision phase of $i$, or if
    \item the previous is false, $v$ is unoccupied during the decision phase of $i$ and no agent $b\in A$ decides to occupy $v$ during the decision phase of $i$.
\end{itemize}
In short, an agent views a vertex as healthy if it is meant to access it, and it is able to do so. 

The \protocolOneShort protocol introduces an additional \textit{modification} phase between the decision and action phases. It also assigns to each agent a state, which can be either \textit{delayed} or \textit{on-time}. Every agent is initially on-time and switches their state to delayed as soon as they perform their first delay-$1$ operation. Let $a$ be an agent of $A$. For each $0\leq i\leq \ell$: 

\begin{enumerate}
    \item if the vertex $v=s_{i}(a)$ is healthy during turn $i$, then the modification phase for agent $a$ is empty and $a$ moves according to its current schedule and its decision during the decision phase.
    \item  Otherwise, the vertex $v=s_{i}(a)$ is not healthy during turn $i$. That is, there exists an agent $b\in A$ such that $s_{i}(b)= v$. Then:
    \begin{enumerate}
        \item If $a$ is on-time, they perform a delay-$1$ operation during the modification phase, stay put during the action phase and are marked as delayed for the rest of the schedule. From this point onward $a$ follows the updated schedule $\sigma'$.
        \item If $a$ is delayed and $b$ is on-time, then $a$ moves according to $\sigma'$ (and $b$ performs a delay-$1$ operation).
        \item If both $a$ and $b$ are delayed, an agent is chosen arbitrarily among the agents $a$ and $b$ to perform an additional delay-$1$ operation, while the other agent moves according to $\sigma'$.
    \end{enumerate}
\end{enumerate}

Our main result is that following the \protocolOneShort protocol results in a feasible schedule which, moreover, is efficient in terms of makespan. 

\begin{theorem}\label{thm:delay-one}
    Let $\mathcal{I}=\langle G, A, s_0, t,\ell\rangle$ be an instance of \MAPFShort{} and $\sigma$ be feasible schedule of $\mathcal{I}$ of length $\ell$. Assume that an adversary forces an agent to perform a malfunction-$1$ operation. If every agent follows the \protocolOneShort protocol, then the resulting schedule $\sigma'$ is feasible, includes the malfunction-$1$ operation imposed by the adversary and is of length $\ell+1$. 
\end{theorem}
\begin{proof}
We will first show that the resulting schedule $\sigma'=(s'_0,\dots,s'_\scheduleLength)$ is feasible. That is, $\sigma'$ is collision-free and $s_\scheduleLength(a)=t(a)$ for every $a\in A$. 

Let us assume that the adversary decides to force the malfunction-$1$ operation upon an agent $a\in A$ during the turn $i$, and let $\sigma'$ be the updated schedule. Note that $a$ is marked as delayed for any turn after this point. We consider what happens after the turn $i$. Let $j\geq i$ be the first turn such that there exists an agent $b\in A$ with $s'_j(b)=s'_j(a)=u$. We know that $a$ is necessarily one of the agents that would participate in such a collision. Indeed $j$ is the first turn after $i$ where a collision occurs, $\sigma$ is feasible, and since the agents are following the \protocolOneShort protocol, no agent has delayed-$1$ except for $a$. We distinguish two cases: 
\begin{itemize}
    \item $j=i$. Then $s_{i-1}(a)=u$, $s_i(a)\neq u$ and $s_i(b)=u$, but $a$ was forced to delay-$1$ by the adversary during the decision phase of turn $i$. Then $u$ is unhealthy during turn $i$, and, according to \protocolOneShort, agent $b$ performs a delay-$1$ operation (since they are on-time). This is a contradiction. 
    \item $j>i$. Since $b$ has not performed any delay-$1$ operations, it is on-time and it follows the initial schedule $\sigma$. So, during the modification phase of turn $j$, the agent $b$ notices that although $s_j(u)$ should be empty, the agent $a$ decided to move to $u$ during the decision phase of turn $j$. Thus, $b$ does not view $u$ as a healthy vertex during turn $j$ and performs a delay-$1$ operation during the modification phase of turn $j$. This is a contradiction.  
\end{itemize}
In any case, it is shown that the turn $j$ doesn't exist. Thus, there is no collision in the schedule $\sigma'$. Also, and by the definition of the delay-$1$ operation, it follows that $\sigma'_\scheduleLength(a)=t(a)$ for every $a\in A$.

We will now show that $\scheduleLength=\ell+1$, for which it suffices showing that $\scheduleLength\ \leq \ell+1$. To achieve this, we will show that no agent will perform more than one delay-$1$ operations. Assume that this claim is false, let $a\in A$ be the first agent that performs a delay-$1$ operation while they are already delayed, and lets say this happens during turn $i$. Observe first that in this case, the adversary has already done their move and forced one vertex to malfunction-$1$, as otherwise no agent would be delayed. Lets analyze this situation. According to the \protocolOneShort protocol, they only way for $a$ to perform a delay-$1$ operation while already marked as delayed, is if there also exists another agent $b\in A$ that is delayed, and both $a$ and $b$ want to access the same vertex at turn $i$. Since the turn $i$ is the first time that this happens, both agents $a$ and $b$ have only been delayed once. That is, $s'_i(a)=s_{i-1}(a)$ and $s'_i(b)=s_{i-1}(b)$. Since they both want to access the same vertex at turn $i$, it follows that $s_{i-1}(a)=s_{i-1}(b)$, which contradicts the feasibility of $\sigma$. 
\end{proof}
Before we move on, note that the overall delay introduced by the agents following the \protocolOneShort protocol is, in essence, optimal. Indeed, it suffices for the adversary to force a malfunction-$1$ operation on an agent whose starting vertex is a distance $\ell$ from its target. Then the resulting schedule will necessarily be of length at least $\ell+1$. 

\paragraph{Upgrading \protocolOneShort.}
The \protocolOneShort protocol deals with the case where the adversary has the potential of forcing only one malfunction upon the agents. We can use this protocol as a building block to construct the \protocolTwo (\protocolTwoShort) protocol, which deals with the case where the adversary can impose $k$ malfunction-$1$ operations on the agents of $A$. We also give to the adversary the power to decide how to distribute these $k$ malfunctions; all of them can be forced upon one single agent, spread out among $k$ agents, or anything in-between. 

The main point where \protocolTwoShort upgrades upon \protocolOneShort is that it assigns a more refined set of states to the agents. In particular, each agents $a\in A$ has a function $d_a:[\scheduleLength]\rightarrow \{0,\dots,k\}$ which returns how many delay-$1$ operations the agent has performed up to the requested turn. That is, $d_a(0)=0$ for every $a\in A$, and if $d_a(i)=j$ for some turn $i$ and $a$ delays-$1$ in turn $i$, then $d_a(i+1)=j+1$. Then, when two agents want to access the same vertex at the same turn, we give the priority to the agent with the smaller value of $d$ and have the other agent delay-$1$. 

The main issue with the \protocolTwoShort protocol is that it would also require an upgraded version of the notion of healthy vertex. In a sense, if an agent wanted to safely access a vertex $u$, then they should check the states of all vertices that are at distance $k$ from $u$. We believe that it would be unrealistic to require such computations from the agents. This leads us to the protocol proposed in the next section.





\iflong

\section{Check Counter Before Moving}

In this section we consider a decision making protocol that tries to rectify the shortcomings of \protocolOneShort{} by deploying a mechanism on each vertex that counts the number of agents that have occupied this vertex up to this turn, including any agent that is currently occupying it.

\paragraph{The \protocolThreeShort{} protocol.}
Let $\mathcal{I}=\langle G, A, s_0, t, \ell\rangle$ be an instance of \MAPFShort{} and $\sigma$ be a feasible schedule of $\mathcal{I}$. 
Before we present the \protocolThree (\protocolThreeShort), we need to formally introduce the following notion. First, each vertex $v\in V(G)$ has its counting function $c_v:[\scheduleLength] \rightarrow \{0,\ldots,k\}$, where $\scheduleLength$ is the length of $\sigma$ and $k=|A|$. For a turn $i$, the function $c_v(i)$ returns the number of agents that have passed through $v$ by (and including) turn $i$. Also, each agent $a$ contains a list $(l^a(s_0),l^a(s_1),\dots,l^a(s_\scheduleLength))$ where for each $i\in[\scheduleLength]$, the element $l^a(s_i)$ is equal to $c_{s_i(a)}(i)$, i.e., how many agents have passed through $s_i(a)$ by turn $i$. 

In the \protocolThreeShort protocol, each agent follows its schedule while respecting the counting of the destination vertex as well as any delays introduced. We say that the agent $a\in A$ \textit{respects} the schedule $\sigma=(s_0(a),\dots,s_\scheduleLength(a))$ during a turn $i$ if:
\begin{itemize}
    \item $s_i(a)=v$,
    \item $s_{i+1}(a)=u$ with $l^a(s_{i+1})=c$,
    \item $c_u(i+1)=c$,
    \item $u$ is free at the beginning of turn $i+1$.   
\end{itemize}
If the above are satisfied, then $a$ moves to $u$. 
In any other case, i.e., if the counter of $u$ has not reached the required number or $u$ is occupied at turn $i+1$, then $a$ performs a delay-$1$ operation.

Our main result here is that following the \protocolThreeShort protocol does not introduce more delays than the number of malfunctions that occur.

\begin{theorem} \label{thm:protocol2}
    Let $\mathcal{I}=\langle G, A, s_0, t,\ell\rangle$ be an instance of \MAPFShort{} and $\sigma$ be feasible schedule of $\mathcal{I}$ of length $\ell$. Assume that an adversary forces $k$ malfunction-$1$ operations upon the agents of $A$. If every agent follows the \protocolThreeShort protocol, then the resulting schedule is feasible, includes the malfunctions imposed by the adversary and is of length $\ell+k$. 
\end{theorem}

\begin{proof}
    Assume that we are in any the start of turn $i < \ell+k$ and that there are at most $k'\le k$ malfunction-$1$ operations that have been imposed, including the ones announced during the start of this turn. It suffices to show that any agent that needs to delay-$1$ in turn $i$ has not already delayed-$1$ more than $k'$ times (including the one just announced). 

    We begin with the case $k'=1$. Note that according to \protocolThreeShort{} no delays are performed before the first malfunctions. We deal with this case by induction on $i$.

    \textbf{Base Case:} Let $i$ be the turn when the adversary forces its one malfunction-$1$.
    Until $i$, no agent was delayed. Thus, the only reason an agent would delay-$1$ during the turn $i$ is if it is the one upon which the malfunction is forced. Therefore, no agent has delayed-$1$ more than once.

    \textbf{Induction Hypothesis:} Assume that after $i$ turns, no agent has introduce a second delay.

    \textbf{Induction Step:} We will show that if any agent introduce a delay after $i+1$ turns, then this is its first delay-$1$. 
    In particular, we will show that any agent that has already delayed-$1$ has no reason to delay-$1$ once more. Consider such an agent. Since $k'=1$, there is no malfunctioning agent this turn. Therefore, any agent $a$ delays-$1$ only if it is planning to move to a vertex $u$ that is occupied by an agent $b$ who is not planning to move this turn.
    Note that $b$ must already have delayed-$1$ as, otherwise, in the original schedule $\sigma$ the agents $a$ and $b$ would occupy $u$ in the same turn, contradicting the feasibility of $\sigma$. 
    Also, since both $a$ and $b$ have delayed-$1$ and $a$ is supposed to move this turn, $b$ is also supposed to move, otherwise they would collide according to $\sigma$. 
    Therefore, we can assume that $b$ is supposed to move to a vertex occupied by an agent that is not planning to move this turn. By applying the same arguments, we can create a sequence of agents that have already delayed-$1$ and cannot move. Since we do not have infinite agents this means that we have created a cycle (i.e. we go through all the agents until we reach $a$ once more). This contradicts the assumption that none of them can move, finishing the proof for the case $k'=1$. 

    We will now prove the claim for a general $k'\leq k$. We show this by induction on $k'$. Note that the base case is already treated above. The new induction hypothesis is as follows:
    
    \textbf{Induction Hypothesis:} If up to (and including) turn $i$ there are at most $k'$ malfunction-$1$ operations that have been forced, then no agent needs to perform more than $k'$ delay-$1$ operations.

    \textbf{Induction Step:} We will prove that, if up to (and including) turn $i$ there are at most $k'+1$ malfunction-$1$ that have been forced, then no agent needs to introduce more than $k'+1$ delay-$1$ operations.

    Let $j$ be the turn when the last malfunction-$1$ operations were performed, and assume there were $k''$ many of them on turn $j$. By the induction hypothesis, up to turn $j-1$, no agents had delayed-$1$ more than $\kappa = k'-k''+1$ times. We will show that for any turn $j'\ge j$, no agent will perform more than $\kappa+1$ delay-$1$ operations.
    We prove this claim by induction on the turns $j'\ge j$.
        
    \textbf{Base Case:} $j'=j$:        
    Since until $j-1$, no more than $\kappa$ delay-$1$ operations have been performed, regardless of which agents needs to delay-$1$, no agent performs more than $\kappa+1$ delay-$1$ operations by turn $j'$.
        
    \textbf{Induction Hypothesis:} Assume that, up to turn $i'\ge j$, no agent has performed more than $\kappa+1$ delay-$1$ operations.

    \textbf{Induction Step:} We will prove that in turn $i'+1$ no agent needs to delay-$1$ more than $\kappa+1$ times. To do so, consider an agent $a$ that has performed $\kappa+1$ delay-$1$ operations up until turn $i$ and assume that $a$ needs to delay-$1$ during the turn $i'+1$. Since there is no new malfunctioning agent, this happen only if, based on the original schedule $\sigma$, the agent $a$ needs to move to a vertex $u$ occupied by an agent $b$ and $b$ want to stay on $u$ during the turn $i'+1$. 
    It is easy to see that $b$ has also performed $\kappa+1$ delay-$1$ operations up to turn $i$. Indeed, if this is not true, then $b$ should visit $u$ after the agent $a$ (according to $\sigma$), or collide with $a$ on $u$. This contradicts the fact that $b$ is currently occupying $u$, which means that the counter in $u$ was the appropriate one when $b$ moved on $u$ (which is impossible without $a$ occupying $u$ before $b$). By the feasibility of $\sigma$, the agent $b$ is supposed to move out of $u$ during the turn $i'+1$, as otherwise it would collide with $a$ in $\sigma$. Therefore, there exists an agent $\gamma$ that occupies the vertex $w$ that $b$ wants to move to at turn $i'+1$, and $\gamma$ is not planning to move out of $w$, forcing $b$ to stay put. By repeating this argument we either run out of agents or create a cycle (we end up to the agent $a$). In the first case, no agent has any reason to delay-$1$. The same is true for the second case as none of the agents appearing in this cycle are malfunctioning. Therefore all of them can move without delaying. This also finishes the induction for the general case of $k'\leq k$.

    Thus, for any $k$ malfunction-$1$ operations that are forced upon the agents, no agent needs to delay-$1$ more than $k$ times. Therefore, all the agents reach their destinations by turn $\ell+k$ in the latest. 
\end{proof}
  
\fi
\ifshort
\input{communication_via_counters-short}
\fi

\section{Conclusion}
In this paper, we formally model the \MAPFShort{} problem where agents are prone to malfunction. The main take-away message is that adapting a given schedule in a centralized manner is infeasible. Thus, we need different approaches. In this work, we propose two distributed protocols that produce feasible schedules of the best possible makespan, assuming a worse case scenario where a powerful adversary is forcing the malfunctions. However, this paper also serves as an open invitation for different approaches to be explored. Is there some heuristic that circumvents the hardness shown in Theorem~\ref{thm:keeping:the:same:makespan:is:NP:hard}? What if the malfunctions were not performed by an adversary, but followed some random distribution? These are just two of the questions that arise from our work. 

\bibliographystyle{abbrv}
\bibliography{biblio}

\begin{thebibliography}{10}

\bibitem{amigoni2022}
F.~Amigoni, D.~Azzalini, N.~Basilico, B.~Flammini, et~al.
\newblock Mapf and mapd: Recent developments and future directions.
\newblock In {\em Proceedings of the AI* IA Workshop on Artificial Intelligence
  and Robotics (AIRO)}, pages 1--6, 2022.

\bibitem{AtzmonSFSK20}
D.~Atzmon, R.~Stern, A.~Felner, N.~R. Sturtevant, and S.~Koenig.
\newblock Probabilistic robust multi-agent path finding.
\newblock In {\em Proceedings of the International Conference on Automated
  Planning and Scheduling}, volume~30, pages 29--37, 2020.

\bibitem{AtzmonSFWBZ20}
D.~Atzmon, R.~Stern, A.~Felner, G.~Wagner, R.~Bart{\'{a}}k, and N.~Zhou.
\newblock Robust multi-agent path finding and executing.
\newblock {\em J. Artif. Intell. Res.}, 67:549--579, 2020.

\bibitem{Bonet10}
B.~Bonet.
\newblock Conformant plans and beyond: Principles and complexity.
\newblock {\em Artificial Intelligence}, 174(3):245--269, 2010.

\bibitem{BredereckCKLN20}
R.~Bredereck, J.~Chen, D.~Knop, J.~Luo, and R.~Niedermeier.
\newblock Adapting stable matchings to evolving preferences.
\newblock In {\em The Thirty-Fourth {AAAI} Conference on Artificial
  Intelligence, {AAAI} 2020, The Thirty-Second Innovative Applications of
  Artificial Intelligence Conference, {IAAI} 2020, The Tenth {AAAI} Symposium
  on Educational Advances in Artificial Intelligence, {EAAI} 2020, New York,
  NY, USA, February 7-12, 2020}, pages 1830--1837. {AAAI} Press, 2020.

\bibitem{chaplick2014contact}
S.~Chaplick, P.~Dorbec, J.~Kratochv{\'\i}l, M.~Montassier, and J.~Stacho.
\newblock Contact representations of planar graphs: Extending a partial
  representation is hard.
\newblock In {\em International Workshop on Graph-Theoretic Concepts in
  Computer Science}, pages 139--151. Springer, 2014.

\bibitem{chaplick2018partial}
S.~Chaplick, G.~Gu{\'s}piel, G.~Gutowski, T.~Krawczyk, and G.~Liotta.
\newblock The partial visibility representation extension problem.
\newblock {\em Algorithmica}, 80(8):2286--2323, 2018.

\bibitem{CimattiRB04}
A.~Cimatti, M.~Roveri, and P.~Bertoli.
\newblock Conformant planning via symbolic model checking and heuristic search.
\newblock {\em Artificial Intelligence}, 159(1):127--206, 2004.

\bibitem{FKKMO24}
F.~Fioravantes, D.~Knop, J.~M. Kristan, N.~Melissinos, and M.~Opler.
\newblock Exact algorithms and lowerbounds for multiagent path finding: Power
  of treelike topology.
\newblock In M.~J. Wooldridge, J.~G. Dy, and S.~Natarajan, editors, {\em
  Thirty-Eighth {AAAI} Conference on Artificial Intelligence, {AAAI} 2024,
  Thirty-Sixth Conference on Innovative Applications of Artificial
  Intelligence, {IAAI} 2024, Fourteenth Symposium on Educational Advances in
  Artificial Intelligence, {EAAI} 2014, February 20-27, 2024, Vancouver,
  Canada}, pages 17380--17388. {AAAI} Press, 2024.

\bibitem{FKKMOV25}
F.~Fioravantes, D.~Knop, J.~M. Kristan, N.~Melissinos, M.~Opler, and T.~A. Vu.
\newblock Solving multiagent path finding on highly centralized networks.
\newblock In T.~Walsh, J.~Shah, and Z.~Kolter, editors, {\em AAAI-25, Sponsored
  by the Association for the Advancement of Artificial Intelligence, February
  25 - March 4, 2025, Philadelphia, PA, {USA}}, pages 23186--23193. {AAAI}
  Press, 2025.

\bibitem{GJ75}
M.~R. Garey and D.~S. Johnson.
\newblock Complexity results for multiprocessor scheduling under resource
  constraints.
\newblock {\em SIAM Journal on Computing}, 4(4):397--411, 1975.

\bibitem{garg2002elements}
V.~K. Garg.
\newblock {\em Elements of distributed computing}.
\newblock John Wiley \& Sons, 2002.

\bibitem{HoffmannB06}
J.~Hoffmann and R.~I. Brafman.
\newblock Conformant planning via heuristic forward search: A new approach.
\newblock {\em Artificial Intelligence}, 170(6):507--541, 2006.

\bibitem{KottingerGASL24}
J.~Kottinger, T.~Geft, S.~Almagor, O.~Salzman, and M.~Lahijanian.
\newblock Introducing delays in multi agent path finding.
\newblock In {\em Proceedings of the International Symposium on Combinatorial
  Search}, volume~17, pages 37--45, 2024.

\bibitem{kshemkalyani2011distributed}
A.~D. Kshemkalyani and M.~Singhal.
\newblock {\em Distributed computing: principles, algorithms, and systems}.
\newblock Cambridge University Press, 2011.

\bibitem{Li_Surynek_Felner_Ma_Kumar_Koenig_2019}
J.~Li, P.~Surynek, A.~Felner, H.~Ma, T.~K.~S. Kumar, and S.~Koenig.
\newblock Multi-agent path finding for large agents.
\newblock {\em Proceedings of the AAAI Conference on Artificial Intelligence},
  33(01):7627--7634, Jul. 2019.

\bibitem{ma2022graph}
H.~Ma.
\newblock Graph-based multi-robot path finding and planning.
\newblock {\em Current Robotics Reports}, 3(3):77--84, 2022.

\bibitem{MaKK17}
H.~Ma, T.~S. Kumar, and S.~Koenig.
\newblock Multi-agent path finding with delay probabilities.
\newblock In {\em Proceedings of the AAAI Conference on Artificial
  Intelligence}, volume~31, 2017.

\bibitem{maLM21}
Z.~Ma, Y.~Luo, and H.~Ma.
\newblock Distributed heuristic multi-agent path finding with communication.
\newblock In {\em 2021 IEEE International Conference on Robotics and Automation
  (ICRA)}, pages 8699--8705, 2021.

\bibitem{NebelK95}
B.~Nebel and J.~Koehler.
\newblock Plan reuse versus plan generation: a theoretical and empirical
  analysis.
\newblock {\em Artificial Intelligence}, 76(1):427--454, 1995.
\newblock Planning and Scheduling.

\bibitem{Pianpak19}
P.~Pianpak, T.~C. Son, P.~O. Toups~Dugas, and W.~Yeoh.
\newblock A distributed solver for multi-agent path finding problems.
\newblock In {\em Proceedings of the First International Conference on
  Distributed Artificial Intelligence}, DAI '19, New York, NY, USA, 2019.
  Association for Computing Machinery.

\bibitem{SharonSFS15}
G.~Sharon, R.~Stern, A.~Felner, and N.~R. Sturtevant.
\newblock Conflict-based search for optimal multi-agent pathfinding.
\newblock {\em Artif. Intell.}, 219:40--66, 2015.

\bibitem{stern2019}
R.~Stern.
\newblock Multi-agent path finding--an overview.
\newblock {\em Artificial intelligence: 5th RAAI summer school, dolgoprudny,
  Russia, July 4--7, 2019, tutorial lectures}, pages 96--115, 2019.

\end{thebibliography}

\end{document}